\documentclass[11pt]{article}

\usepackage{ulem}
\usepackage{authblk}
\newcommand{\blind}{0}

\usepackage[margin=1in]{geometry}
\usepackage{setspace}
\usepackage{caption} 
\usepackage{amsthm,amsmath,amsfonts,amssymb,bm}
\usepackage[authoryear]{natbib}
\usepackage{graphicx,psfrag,epsf}
\usepackage{multicol}
\usepackage{color,hyperref,xcolor}
\hypersetup{colorlinks=true,urlcolor=purple,citecolor=blue}
\usepackage{cleveref}
\usepackage{longtable}
\usepackage{natbib}

\def\E{\mathbb{E}}
\def\R{\mathbb{R}}
\def\I{\mathbf{I}}

\def\P{\mathbb{P}}
\def\X{\bm{X}}
\def\L{\bm{L}}
\def\Y{\bm{Y}}
\def\N{\mathcal{N}}
\def\F{\mathcal{F}}

\newcommand{\boldeps}{\bm{\varepsilon}}
\newcommand{\boldalpha}{\bm{\alpha}}
\newcommand{\boldmu}{\bm{\mu}}
\newcommand{\boldp}{\bm{p}}

\newcommand{\abs}[1]{\left\lvert #1 \right\rvert}
\newcommand{\norm}[1]{\left\| #1 \right\|}
\newcommand{\cov}{\mathrm{Cov}}
\newcommand{\var}{\mathrm{Var}}

\newcommand{\given}{\hspace{1pt}\vert\hspace{1pt}}

\newcommand{\multi}{\mathrm{Multinomial}}

\numberwithin{equation}{section}  

\newtheoremstyle{general}
{3mm} 
{3mm} 
{\it} 
{} 
{\bfseries} 
{.} 
{.5em} 
{} 

\theoremstyle{general}

\newtheorem{lemma}{Lemma}
\newtheorem{theorem}{Theorem}
\newtheorem{proposition}{Proposition}

\newtheorem{assumption}{Assumption}
\newtheorem{remark}{Remark}

\makeatletter
\renewenvironment{proof}[1][\proofname]{\par
	\pushQED{\qed}%
	\normalfont \topsep6\p@\@plus6\p@\relax
	\trivlist
	\item\relax{
		\bfseries
		#1\@addpunct{.}}\hspace\labelsep\ignorespaces
}{%
	\popQED\endtrivlist\@endpefalse
}
\makeatother

\begin{document}
	
	{\large 
		\noindent {\bf Title:} Forecasting Elections from Partial Information Using a Bayesian Model for a Multinomial Sequence of Data
		
		\vspace{0.1in}
		\noindent {\bf Running title:} A Bayesian model to forecast elections from partial information
		
		\vspace{0.1in}
		\noindent {\bf Authors:} Soudeep Deb, Rishideep Roy, Shubhabrata Das
		
		\vspace{0.1in}
		\noindent {\bf Affiliation:} 
		
		\noindent Indian Institute of Management Bangalore \\ Bannerghatta Main Road, Bangalore, KA 560076, India. \\
		School of Mathematics, Statistics and Actuarial Science, University of Essex \\ Wivenhoe Park, Colchester, Essex, CO4 3SQ, UK. 
		
		\vspace{0.1in}
		\noindent {\bf Corresponding author:} 
		
		\noindent Soudeep Deb \\
		Indian Institute of Management Bangalore \\ Bannerghatta Main Road, Bangalore, KA 560076, India. \\
		Email: soudeep@iimb.ac.in.
	}
	
	\vspace{0.1in}
	\noindent {\bf Data availability statement:}
	
	\noindent Two datasets are used in this paper. They have been downloaded from \url{https://eci.gov.in/files/file/12787-bihar-legislative-election-2020/} and \url{https://dataverse.harvard.edu/file.xhtml?fileId=4788675&version=8.0}.
	
	\vspace{0.1in}
	\noindent {\bf Declaration of interest:}
	
	\noindent The authors declare no conflict of interest.
	
	\newpage
	\setcounter{page}{1}
	
	\def\spacingset#1{\renewcommand{\baselinestretch}%
		{#1}\small\normalsize} \spacingset{1}

	\if0\blind
	{
		\title{\bf Forecasting Elections from Partial Information Using a Bayesian Model for a Multinomial Sequence of Data}
	\author[1]{Soudeep Deb\thanks{Corresponding author. Email: soudeep@iimb.ac.in}} 
	\author[2]{Rishideep Roy\thanks{Email: rishideep.roy@essex.ac.uk}}
	\author[3]{Shubhabrata Das\thanks{Email: shubho@iimb.ac.in}\hspace{.2cm}} 
	\affil[1,3]{Indian Institute of Management Bangalore \\ Bannerghatta Main Road, Bangalore, KA 560076, India.}
	\affil[2]{School of Mathematics, Statistics and Actuarial Science\\University of Essex\\ Wivenhoe Park, Colchester, Essex CO4 3SQ, UK.}
	\maketitle
} \fi

\if1\blind
{
	\bigskip
	\bigskip
	\bigskip
	\begin{center}
		\title{\bf Forecasting Elections from Partial Information Using a Bayesian Model for a Multinomial Sequence of Data}
\end{center}
\medskip
} \fi

\begin{abstract}
Predicting the winner of an election is of importance to multiple stakeholders. To formulate the problem, we consider an independent sequence of categorical data with a finite number of possible outcomes in each. The data is assumed to be observed in batches, each of which is based on a large number of such trials and can be modelled via multinomial distributions. We postulate that the multinomial probabilities of the categories vary randomly depending on batches. The challenge is to predict accurately on cumulative data based on data up to a few batches as early as possible. 
On the theoretical front, we first derive sufficient conditions of asymptotic normality of the estimates of the multinomial cell probabilities and present corresponding suitable transformations. Then, in a Bayesian framework, we consider hierarchical priors using multivariate normal and inverse Wishart distributions and establish the posterior convergence. The desired inference is arrived at using these results and ensuing Gibbs sampling. The methodology is  demonstrated with election data from two different settings --- one from India and the other from the United States of America. Additional insights of the effectiveness of the proposed methodology are attained through a simulation study.  
\end{abstract}

{\it Keywords: Election data, Gibbs sampling, Hierarchical priors, Posterior convergence, Forecasting from partial information}.

\newpage

\section{Introduction}
\label{sec:intro}

Voter models have been of interest for a long time.  With the strengthening of democracy and the role of media, statistical models for the prediction of elections have been on the rise around the world. Elections are studied at different stages -- starting with pre-poll predictions, going on to post-poll predictions, and also calling the election before the final result is declared as the counting of ballots is in progress. The current work is focuses on the last context, where the counting of results is publicly disclosed in batches (or rounds). Stemmed from public interest, this often leads to competition among the media houses ``to call the election", i.e.\ to predict the outcome correctly and as early as possible.

Using the exit poll data has been a popular approach in this context. It however has its own pitfalls, as has been shown in many papers (see, e.g.\ \cite{stiers2018affect}). Because of the inadequacy of the exit poll data, Associated Press (AP) designed its methodology for both survey and forecasting (\cite{APnews}). Interestingly, they still do not call a closely contested race unless the trailing candidate cannot mathematically secure a victory. Another common approach is to extract relevant information from social media, such as Twitter or Google Trends, and to use that to predict the winner of an election. \cite{o2010tweets} and \cite{tumasjan2010predicting} are two notable works in this regard. A major criticism behind this approach is that intentional bias in the social media posts often lead to error-prone conclusions in the forecasting problems (\cite{anuta2017election}). Further, \cite{haq2020survey} provides a great review of the existing literature on how to call an election and one can see that most of the methods are either based on data from social media or are ad-hoc and devoid of sound statistical principles. To that end, it would be of paramount importance to develop a good forecasting methodology based only on the available information of votes secured by the candidates in the completed rounds of counting. This can be treated as a collection of multinomial data where the number of categories is equal to the number of candidates and the random fluctuations in cell probabilities point to the randomness of the vote shares in different rounds.

Multinomial distribution is the most common way to model categorical data. In this work, we study the behaviour of a collection of multinomial distributions with random fluctuations in cell (or `category', as we interchangeably use the two terminology) probabilities. We consider  situations where this  multinomial data is observed in batches, with the number of trials in each batch being possibly different. The cell probabilities for every multinomial random variable in the data are assumed to be additive combinations of a fixed component which is constant across batches, and a random perturbation. These perturbations for different batches are taken to be independent. Before delving deeper into the model and related results, we present an interesting application where the proposed setup can be utilized effectively. 

\subsection{Motivating Example}

This study is largely motivated by the forecasting results in political elections during the process of counting of votes; in particular, we look at one such instance in India -- the electoral data from the legislative assembly election held in Bihar (a state in India), during October-November of 2020. There were two major alliances -- the National Democratic Alliance (NDA), and the {\it Mahagathbandhan} (MGB). The counting of votes was declared in different rounds, and the contest being very close, there were widespread speculated forecasts in various media for nearly 20 hours since the counting began, until the final results were eventually known. 


To illustrate the challenges of this forecasting exercise, let us consider Hilsa and Parbatta, two of the 243 constituencies in Bihar where the contests were among the closest. In both of them, like in most other constituencies in the state, the fight was primarily confined between NDA and MGB. For both Hilsa and Parbatta, the counting of votes was completed in 33 rounds. The number of votes counted in the different rounds varied greatly (between two to seven thousand, except the last one or two rounds which typically had less votes). In \Cref{fig:bihar-graph-intro}, we depict the number of votes by which the candidates led (roundwise numbers are given on the left panel, cumulative numbers are given in the right panel). For instance, we observe that among the first 10 rounds of counting for Hilsa, the NDA candidate led in the sixth, eighth and ninth rounds, while the MGB candidate led in counting of votes in the remaining seven rounds. Considering cumulative vote counts in Hilsa at the end of successive rounds, lead changed once; that was in the final round with NDA winning by meager 12 votes. In contrast, the lead changed as many as five times in Parbatta, viz.\ in the 16th, 20th, 23rd, 26th and 31st rounds, before NDA claimed the victory with a margin of 951 votes.





\begin{figure}[!hbt]
\centering
\includegraphics[width = \textwidth,keepaspectratio]{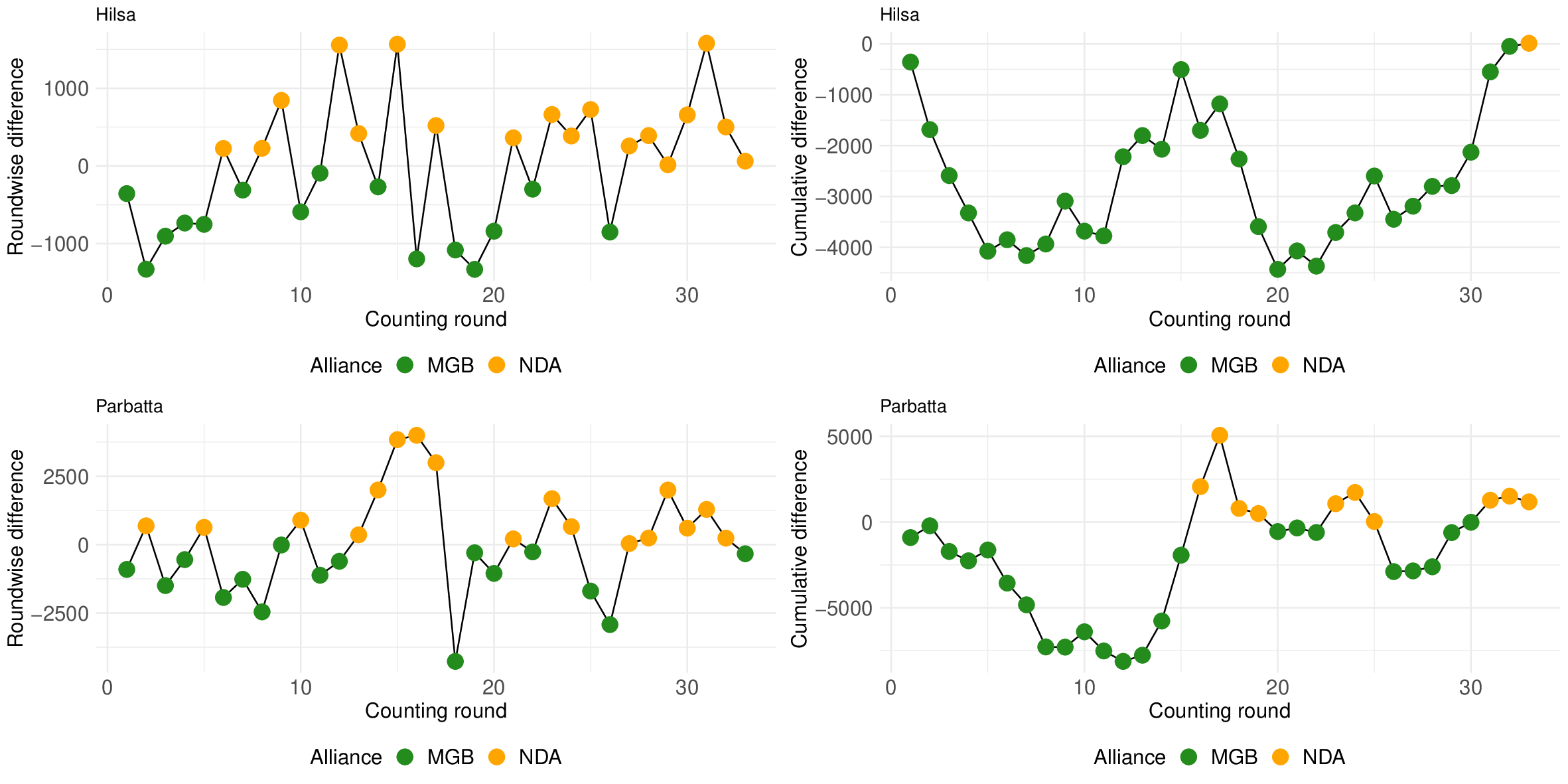}
\caption{Round-wise margin of differences in votes (left) and cumulative margins of votes (right) for the two main alliances in Hilsa and Matihani, two most closely contested constituencies.}
\label{fig:bihar-graph-intro}
\end{figure}

The above clearly demonstrates the difficulty in forecasting the winner (candidate who eventually gets the highest number of votes from all rounds of counting) of individual constituencies, as early (after as few rounds of counting) as possible. Our goal is to construct a method that relies on sound statistical principles to predict the winner for each seat based on trends, thereby making a call about which alliance is going to secure a majority. 

We note that the methodology can be adapted to other political applications as well. For example, in the context of the presidential election of the United States of America (USA), one may consider the data from a few counties to project the winner of the respective state. We illustrate it later in this paper, although in a limited capacity because of not having complete information on the actual timing when the county-wise results were progressively declared. 

\subsection{Related literature and our contribution}

Literature on a non-identical sequence of categorical data is somewhat limited. In the binomial setting (i.e.\ with two categories of outcome), \cite{le1960approximation} introduced the concept of independent but non-identical random indicators, which was later discussed and applied in a few other studies (see, e.g.\ \cite{volkova1996refinement}, \cite{hong2009prediction}, \cite{fernandez2010closed}). However, in these papers, the cell probabilities are considered to be different but non-random. Introducing the randomness in the probabilities makes it more attractive and suitable for many real life applications. To that end, \cite{cargnoni1997bayesian} proposed a conditionally Gaussian dynamic model for multinomial time series data and used it to forecast the flux of students in the Italian educational system. This model has later been adopted and modified in different capacities, see \cite{landim2000dynamic}, \cite{terui2010finding}, \cite{quinn2010analyze}, \cite{terui2014multivariate} for example. 

In this paper, we develop an appropriate framework for the above-mentioned setup under mild assumptions on the behaviour of the random perturbations. The model is implemented through Bayesian techniques in a setting similar to  \cite{cargnoni1997bayesian}. We complement their work by providing relevant theoretical results of the proposed method and describe how the model can be used to forecast and infer about different features of the entire data. Next, in line with the motivations discussed earlier, in addition to simulation, we demonstrate with two 
contextual applications of our proposed methodology. First one is related to election prediction where we use the data from Bihar 
legislative assembly election held in 2020. 
The second 
is for US election; although in absence of availability of  sequencing of the actual data, we consider different hypothetical possibilities. 
To the best of our knowledge, in the context of these applications, the present work is the first attempt to utilize a randomly varying probability structure for a sequence of non-identical multinomial distributions. We emphasize that the goal of this study is to forecast aspects of the final outcome rather than inference on the multinomial cell probabilities, although the latter can also be easily done through the proposed framework.

\subsection{Organization}

The rest of the paper is arranged as follows. In \Cref{subsec:preliminaries}, we introduce the preliminary notations, and in \Cref{subsec:setup-main}, we describe the model in a formal statistical framework, and state the main results. The proofs of the theorems are provided in \Cref{sec:proofs}. In \Cref{subsec:gibbs}, we discuss the estimation of the posterior distribution of the model parameters via Gibbs sampling. The prediction aspects of the model are described in \Cref{subsec:prediction}. \Cref{sec:simulation} discusses simulation results which demonstrate the suitability of the model in the broad context. Real data application in an Indian context is presented in \Cref{sec:Bihar_election}. We also present a hypothetical case study from American politics in \Cref{sec:usa_application}. Finally, we conclude with a summary, some related comments and ways to extend the current work in \Cref{sec:conclusion}.


\section{Methods}
\label{sec:methods}

\subsection{Preliminaries}
\label{subsec:preliminaries}

In this section, we describe the notations and assumptions, and then for better understanding of the reader, draw necessary connections to the previously discussed examples.

At the core of the framework, we have a sequence of independent trials, each of which results in one among the possible $C$ categories of outcomes. The outcomes of the trials are recorded in aggregated form, which we refer to as batches (or rounds). We assume that the cell probabilities remain the same within a batch, but vary randomly across different batches. Hence, the observed data from an individual batch can be modelled by the standard multinomial distribution; and collectively we have a sequence of independent but non-identical multinomial random variables. 

We adopt the following notations. Altogether, the multinomial data is available in $K$ batches.  For each batch, the number of trials is given by $n_j$, $1\leqslant j \leqslant K$, and each $n_j$ is assumed to be large. Let $N_j=\sum_{i=1}^j n_i$ denote the cumulative number of trials up to the $j^{th}$ batch. For simplicity, let $N=N_K$, the total number of trials in all the batches. We use $\X_{j}=(X_{1j}, X_{2j}, \ldots,X_{Cj})^T$ to denote the observed counts for the $C$ different categories in case of the $j^{th}$ multinomial variable in the data. Also, let $\Y_{j}=(Y_{1j}, Y_{2j},\ldots, Y_{Cj})^T$ be the cumulative counts till the $j^{th}$ batch. Clearly,
\begin{equation}
\label{eq:definition_y}
\sum_{i=1}^j X_{ci} = Y_{cj}, \quad  \text{for }c \in \{1,2,\ldots, C\}.
\end{equation}

Let $p_{1j}$, $p_{2j}$,\ldots, $p_{Cj}$ denote the probabilities for the $C$ categories for each of the independent trials in the $j^{th}$ batch. We focus on $\{ p_{1j}, p_{2j}, \ldots, p_{(C-1)j}\}$, since the probability of the last category is then automatically defined. Our objective is to estimate the probability distribution, and subsequently various properties, of $(h(\Y_{l})\mid \F_j)$  for $l>j$, where $h$ is some suitable measurable function and $\F_j$ is the sigma-field generated by the data up to the $j^{th}$ batch.

If we focus on the aforementioned voting context of Bihar (or, the USA), $C$ stands for the number of candidates contesting in a constituency (or, a state), and each voter casts a vote in favour of exactly one of these candidates. Here, $n_j$ represents the number of voters whose votes are counted in the $j^{th}$ round (or, the $j^{th}$ county), while $N_j$ denotes the corresponding cumulative number of votes. Similarly, $\X_{j}$ shows the number of votes received by the candidates in the $j^{th}$ round (or, the $j^{th}$ county), and $\Y_{j}$ represents the corresponding vector of cumulative votes. We can use our method to find the probability of winning for the $i^{th}$ candidate, given the counted votes from the first $j$ number of rounds (or, counties). In this case, the measurable function of interest is $h_{1,i}(\Y_K)=\mathbb{I}(Y_{iK} > \max_{l\ne i} Y_{lK})$, where $\mathbb{I}(\cdot)$ denotes the indicator function. Similarly, using the function $h_2(\Y_K)=\Y_K^{(1)} - \Y_K^{(2)}$, where $\Y_K^{(1)}$ and $\Y_K^{(2)}$ respectively denote the first and second order statistics of $\Y_K$, we can predict the margin of victory as well.  


\subsection{Model framework and main results}
\label{subsec:setup-main}

We propose to consider the multinomial cell probabilities as random variables, having a fixed part and a randomly varying part. Thus, for $c=1,2,\ldots,C$, $p_{cj}$ is written as $p_c+\varepsilon_{cj}$, where $p_c$ is the constant preference component that remains the same across the batches for category $c$, and $\varepsilon_{cj}$'s, for $c \in \{1,2,\ldots, C \}$, are zero-mean random perturbations  that model possible fluctuations in the preference probabilities for the categories across the batches. We enforce $\varepsilon_{Cj}=-\sum_{c=1}^{C-1}\varepsilon_{cj}$. In line with the earlier notations, we use $\boldp = (p_1,p_2,\ldots,p_C)^T$, $\tilde{\boldp}=(p_1,p_2,\ldots,p_{C-1})^T$, $\boldp_j=(p_{1j},p_{2j},\ldots, p_{Cj})^T$ and $\boldeps_j=(\varepsilon_{1j},\varepsilon_{2j},\ldots, \varepsilon_{(C-1)j})^T$, for convenience. Note that the covariances between the components of $\boldeps_j$'s are likely to be negative due to the structure of the multinomial distribution. The randomness of $\boldeps_{j}$'s need to be also carefully modelled so as to ensure that the category probabilities $p_{ij}$'s lie in the interval $[0,1]$. We shall use $\mathcal{S}$ to denote the possible set of values for $\boldeps_{j}$. 

Let $\cov(\cdot,\cdot)$ denote the variance-covariance matrix of two random variables and $\N_r(\bm{\theta},\Psi)$ denote a $r$-variate normal distribution with mean $\bm{\theta}$ and dispersion matrix $\Psi$. Throughout this article, $\bm{0}$ denotes a vector of all zeroes and $\I$ denotes an identity matrix of appropriate order. Following is a critical assumption that we use throughout the paper.

\begin{assumption}
\label{asmp:exp_var_epsilon}
The random variables $(\boldeps_j)_{1\leqslant j\leqslant K}$, as defined above, are independent of each other and are distributed on $\mathcal{S}$, with $\E(\varepsilon_{cj})=0$ for $1\leqslant c \leqslant C-1$. Further, the density of $\sqrt{n_j} \boldeps_j$ converges uniformly to the density of $\N_{C-1}(\bm 0,\Xi_{\boldeps})$, for some positive definite matrix $\Xi_{\boldeps}$, typically unknown.
\end{assumption}

Then, the proposed model can be written as
\begin{equation}
\label{eq:proposed_model}
\X_{j}\given \boldeps_j \sim \multi\left(n_j,p_1+\varepsilon_{1j},p_2+\varepsilon_{2j},\ldots, p_{C-1}+\varepsilon_{(C-1)j},p_C-\sum_{c=1}^{C-1}\varepsilon_{cj}\right),
\end{equation}
where $(\X_j)_{1\leqslant j\leqslant K}$ are independent of each other and $(\boldeps_j)_{1\leqslant j\leqslant K}$ satisfy \Cref{asmp:exp_var_epsilon}.

We shall adopt a Bayesian framework to implement the above model. Before that, it is imperative to present a couple of frequentist results on the distribution of the votes, as they help us in setting up the problem, and serve as the basis of the framework and the results stated in \Cref{subsec:gibbs}. Below, \Cref{thm:unconditional-x-y} shows the prior distribution of the voting percentages, while \Cref{lem:sin-inverse-clt} helps us in getting the distribution of the transformed prior that we actually use in our Bayesian hierarchical models. 


\begin{proposition}
\label{thm:unconditional-x-y}
For the proposed model (\cref{eq:proposed_model}) the following results are true if \Cref{asmp:exp_var_epsilon} is satisfied. 
\begin{itemize}
	
	\item[(a)] Unconditional first and second order moments of $X_{cj}$, $X_{c'j}$ ($c,c'=1,2,\hdots,C$, $c\neq c'$) are given by
	\begin{equation}
		\label{eq:unconditional-moments-x}
		\E(X_{cj})=n_jp_c, \; \cov(X_{cj},X_{c'j}) = n_j \begin{bmatrix}
			p_c (1-p_c) & -p_c p_{c'} \\
			-p_c p_{c'} & p_{c'}(1-p_{c'})
		\end{bmatrix} + n_j (n_j -1) \cov(\varepsilon_{cj},\varepsilon_{c'j}).
	\end{equation}
	
	\item[(b)] Unconditional first and second order moments of $Y_{cj}$, $Y_{c'j}$ ($c,c'=1,2,\hdots,C$, $c\neq c'$) are given by
	\begin{equation}
		\label{eq:unconditional-moments-y}
		\E(Y_{cj})=N_jp_c, \; \cov(Y_{cj},Y_{c'j}) = N_j \begin{bmatrix}
			p_c (1-p_c) & -p_c p_{c'} \\
			-p_c p_{c'} & p_{c'}(1-p_{c'})
		\end{bmatrix} + \sum_{i=1}^j n_i (n_i -1) \cov(\varepsilon_{ci},\varepsilon_{c'i}).
	\end{equation}
	
	\item[(c)] As $n_j\to\infty$, $\hat \boldp_j = (\hat p_{1j},\hat p_{2j},\hdots,\hat p_{(C-1)j})^T = (X_{1j}/n_j,X_{2j}/n_j,\hdots,X_{(C-1)j}/n_j)^T$ satisfies the following:
	\begin{equation} 
		\label{eq:asymptotic-xnj}
		\sqrt{n_j}\left(\hat \boldp_j - \tilde \boldp \right) \overset{\mathcal{L}}{\rightarrow} \N_{C-1} \left(\bm{0}, \Xi \right),
	\end{equation}
	where $\overset{\mathcal{L}}{\rightarrow}$ denotes convergence in law, and $\Xi=\Xi_p+\Xi_{\boldeps}$. Here, $\Xi_{\boldeps}$ is positive definite with finite norm and 
	\begin{equation}\label{eq:Xi-p}
		\Xi_p = { \begin{bmatrix} p_1 (1-p_1) & -p_1 p_2 & \hdots & -p_1p_{C-1}  \\ -p_2 p_1 & p_2(1-p_2) & \hdots & -p_2p_{C-1} \\ \vdots & \vdots & \ddots & \vdots \\ -p_{C-1}p_1 & -p_{C-1}p_2 & \hdots & p_{C-1}(1-p_{C-1}) \end{bmatrix}}.
	\end{equation}
\end{itemize}
\end{proposition}

\begin{remark}\label{rem:asymptotic-xnj}
In part (c) of \Cref{thm:unconditional-x-y}, if the variances and covariance of all components of $\boldeps_j$ are $o(1/n_j)$ then as $n_j\to\infty$, $\sqrt{n_j}\left(\hat \boldp_j - \tilde \boldp \right) \overset{\mathcal{L}}{\rightarrow} \N_{C-1} \left(\bm{0}, \Xi_p \right)$.
\end{remark}

The above remark includes the scenario of independent and identically distributed multinomials and is therefore of special interest. Detailed proof of \Cref{thm:unconditional-x-y} is deferred to \Cref{sec:proofs}. Note that the asymptotic variance is a function of $\boldp$, thereby motivating us to adapt an appropriate adjustment in the form of using a suitable variance stabilizing transformation. While the standard variance stabilizing transformation is given by the sine-inverse function (\cite{anscombe1948transformation}), there are some better variance stabilizing transformations, as discussed in \cite{yu2009variance}. Based on that work, we define the following modified transformation, having superior relative errors, skewness and kurtosis:

\begin{equation}
\label{eq:sin-inverse-function}
\L_{j} = 
\begin{pmatrix}
	L_{1j} \\
	L_{2j} \\
	\vdots \\
	L_{(C-1)j}
\end{pmatrix}
=
\begin{pmatrix}
	\sin^{-1} {\dfrac{2\hat p_{1j}-1}{1+2a/n_j}} \\
	\sin^{-1} {\dfrac{2\hat p_{2j}-1}{1+2a/n_j}} \\
	\vdots \\
	\sin^{-1} {\dfrac{2\hat p_{(C-1)j}-1}{1+2a/n_j}}
\end{pmatrix},
\end{equation}
where $a$ is a positive constant. Throughout this paper, we consider $a=3/8$, one of the most popular choices in this regard, cf.\ \cite{anscombe1948transformation} and \cite{yu2009variance}. 

\begin{lemma}
\label{lem:sin-inverse-clt}
For every $1\leqslant j\leqslant K$ in the framework of the proposed model (\cref{eq:proposed_model}), under \Cref{asmp:exp_var_epsilon}
\begin{equation}
	\label{eq:asymptotic-Lnj}
	\sqrt{n_j+0.5}\left(\L_j-\boldmu\right) \overset{\mathcal{L}}{\longrightarrow} \N_{C-1} \left( \bm 0, \Sigma  \right),  ~\text{ as } n_j \to \infty,
\end{equation}
where $\boldmu = (\sin^{-1}(2p_1-1),\sin^{-1}(2p_2-1),\ldots,\sin^{-1}(2p_{C-1}-1))$, and $\Sigma$ is a variance-covariance matrix with $\norm{\Sigma}<\infty$.

In particular, if the variances and covariances of $\varepsilon_{1j}, \ldots, \varepsilon_{(C-1)j}$ are $o(1/n_j)$ then $\Sigma$ is a variance-covariance matrix with diagonal entries as $1$, and off-diagonal entries given by $$\Sigma_{cc'}=-\sqrt{\frac{p_c p_{c'}}{(1- p_c)(1- p_{c'})}} \mbox{\rm for } c\neq c' \in \{1,2,\hdots, C-1 \}. $$
\end{lemma}

\begin{proof}
For a function $g : \mathbb{R}^{C-1} \rightarrow \mathbb{R}^{C-1}$ of the form: $$g(x_1, x_2, \hdots, x_{C-1}) = \begin{pmatrix}
	g_1 (x_1) \\
	g_2 (x_2) \\
	\vdots \\
	g_{C-1} (x_{C-1}) 
\end{pmatrix},$$
using  \cref{eq:asymptotic-xnj} and the multivariate delta theorem (see
\cite{cox2005delta} and \cite{ver2012invented} for the theorem and related discussions) we get
\begin{equation}
	\label{eq:multivariate-delta}
	\sqrt{n_j}\left(g(\hat p_{1j},\hat p_{2j},\hdots,\hat p_{(C-1)j}) - g(p_1,p_2,\hdots,p_{C-1})\right) \overset{\mathcal{L}}{\rightarrow} \N_{C-1}\left(\bm 0,\Omega\right),
\end{equation}
where the $(c,c')^{th}$ entry of $\Omega$ is
$$\Omega_{c,c'} = \begin{cases} (g_c'(p_c))^2 \; \Xi_{c,c} & \text{for } c=c', \\ -g_c'(p_c)g_{c'}'(p_{c'})\; \Xi_{c,c'} & \text{for } c\neq c'. \end{cases}$$

Using $g_1(t)=g_2(t)=\cdots=g_{C-1}(t)=\sin^{-1}((2t-1)/(1+2a/n_j))$ in the above, the required result follows. For the special case in the second part, the exact form of $\Xi_p$ from \cref{eq:Xi-p} can be used.
\end{proof}

The above lemma serves as a key component of our method elaborated in the following subsection. Note that $g_1(t)$ is a monotonic function on $[0,1]$ and therefore, we can easily use the above result to make inference about $\tilde \boldp$, and subsequently about $(h(\Y_{l})\mid \F_j)$ for $l>j$.

\subsection{Bayesian estimation}
\label{subsec:gibbs}

In this paper, for implementation of the model, we adopt a Bayesian framework which is advantageous from multiple perspectives. First, we find this to yield more realistic results in several real-life data context. Second, it helps in reducing the computational complexity, especially for a large dataset. Third, this approach has the flexibility to naturally extend to similar problems in presence of covariates. 

Following \cref{eq:asymptotic-Lnj} and the notations discussed above, the proposed model (see \cref{eq:proposed_model}) is equivalent to the following in an asymptotic sense.
\begin{equation}
\L_j \sim \N_{C-1}\left(\boldmu,\frac{\Sigma}{n_j+0.5}\right), \; 1\leqslant j\leqslant K.
\end{equation}

Observe that both $\boldmu$ and $\Sigma$ are unknown parameters and are fixed across different batches. In the Bayesian framework, appropriate priors need to be assigned to these parameters to ensure ``good" behaviour of the posterior distributions. To that end, we consider the following hierarchical structure of the prior distributions. 
\begin{equation}
\label{eq:hierarchical-priors}
\boldmu \sim \N_{C-1}\left(\boldalpha,\Sigma_p\right), \; \Sigma \sim \mbox{Inverse Wishart}(\Psi,\nu), \; \Sigma_p \sim \mbox{Inverse Wishart}(\Psi_p,\nu_p).
\end{equation}

The inverse Wishart prior is the most natural conjugate prior for the covariance matrix, cf.\ \cite{chen1979bayesian}, \cite{haff1980empirical}, \cite{barnard2000modeling}, \cite{champion2003empirical}. The conjugacy property facilitates amalgamation into Markov chain Monte Carlo (MCMC) methods based on Gibbs sampling. This leads to easier simulations from the posterior distribution as well as posterior predictive distribution after each round. On this note, the inverse Wishart distribution, if imposed only on $\Sigma$ with $\Sigma_p$ known, restricts the flexibility in terms of modeling prior knowledge (\cite{hsu2012bayesian}). To address this issue, we apply the above hierarchical modeling as is most commonly used in Bayesian computations. This allows greater flexibility and stability than diffuse priors (\cite{gelman2006prior}). Studies by \cite{kass2006default}, \cite{gelman2006data}, \cite{bouriga2013estimation} discuss in depth how the hierarchical inverse Wishart priors can be used effectively for modeling variance-covariance matrices. 


Before outlining the implementation procedure of the above model, we present an important result which establishes the large sample property of the posterior means of $\boldmu$ and $\Sigma$.

\begin{theorem}
\label{thm:posterior-consistency}
Let $\boldmu_0$ and $\Sigma_0$ be the true underlying mean and true covariance matrix, respectively. Let $\boldmu_{PM}$ and $\Sigma_{PM}$ be the posterior means. Then, for the hierarchical prior distributions given by \cref{eq:hierarchical-priors}, the following are true:   
\begin{equation}\label{eq:postconstmu}
	\lim_{K \rightarrow \infty} \E\left[\Pi_K\left( ||\boldmu_{PM} - \boldmu_0|| > \epsilon \mid \L_1, \L_2, \ldots,\L_K \right)\right] =0,
\end{equation}
\begin{equation}\label{eq:postconstSigma}
	\lim_{K \rightarrow \infty} \E\left[\Pi_K\left( ||\Sigma_{PM} - \Sigma_0|| > \epsilon \mid \L_1, \L_2, \ldots,\L_K \right)\right] =0.
\end{equation}
\end{theorem}


In terms of implementation, it is usually a complicated procedure to find out a closed form for the joint posterior distribution of the parameters. A common approach is to adopt Gibbs sampling which reduces the computational burden significantly. It is an MCMC method to obtain a sequence of realizations from a joint probability distribution. Here, every parameter is updated in an iterative manner using the conditional posterior distributions given other parameters.  We refer to \cite{geman1984stochastic} and \cite{durbin2002simple} for more in-depth readings on Gibbs sampling. 

Recall that $\F_j$ denotes the sigma-field generated by the data up to the $j^{th}$ instance. In addition, below, $\Gamma_k$ stands for a $k$-variate gamma function and $\mbox{tr}(\cdot)$ is the trace of a matrix. 

For $j=1$, we have $\L_{1}\given\boldmu,\Sigma \sim \N_{C-1} \left(\boldmu, \Sigma/(n_1+0.5)\right)$. The joint posterior likelihood is therefore given by 
\begin{align}
f(\boldmu , \Sigma, \Sigma_p \given \F_1) &\propto |\Sigma|^{-1/2} \exp\left[ -\frac{n_1+0.5}{2} (\L_{1} -\boldmu)^T \Sigma^{-1} (\L_{1} -\boldmu)\right]
\nonumber \\
&\times  |\Sigma_p|^{-1/2} \exp\left[ -\frac{1}{2 } (\boldmu -\boldalpha)^T \Sigma_p^{-1}  (\boldmu -\boldalpha)\right]  \nonumber \\
&\times \dfrac{|\Psi|^{\nu/2}}{2^{\nu}\Gamma_2(\nu/2)}|\Sigma|^{-(\nu+C)/2}\exp \left[-\frac{1}{2}\mbox{tr}(\Psi \Sigma^{-1})\right] \nonumber \\
\label{eq:joint_posterior_round1}
&\times \dfrac{|\Psi_p|^{\nu_p/2}}{2^{\nu_p}\Gamma_2(\nu_p/2)}|\Sigma_p|^{-(\nu_p+C)/2}\exp \left[-\frac{1}{2}\mbox{tr}(\Psi_p \Sigma_p^{-1})\right].
\end{align}

In an exact similar way, the posterior likelihood based on the data up to the $j^{th}$ instance can be written as
\begin{align}
f(\boldmu, \Sigma, \Sigma_p \given \F_j) &\propto |\Sigma|^{-j/2} \exp\left[ -\sum_{i=1}^j \frac{n_i+0.5}{2} (\L_i -\boldmu)^T \Sigma^{-1} (\L_i -\boldmu)\right]
\nonumber \\
&\times  |\Sigma_p|^{-1/2} \exp\left[ -\frac{1}{2 } (\boldmu -\boldalpha)^T \Sigma_p^{-1}  (\boldmu -\boldalpha)\right]  \nonumber \\
&\times \dfrac{|\Psi|^{\nu/2}}{2^{\nu}\Gamma_2(\nu/2)}|\Sigma|^{-(\nu+C)/2}\exp \left[-\frac{1}{2}\mbox{tr}(\Psi \Sigma^{-1})\right] \nonumber \\
\label{eq:joint_posterior_general}
&\times \dfrac{|\Psi_p|^{\nu_p/2}}{2^{\nu_p}\Gamma_2(\nu_p/2)}|\Sigma_p|^{-(\nu_p+C)/2}\exp \left[-\frac{1}{2}\mbox{tr}(\Psi_p \Sigma_p^{-1})\right].
\end{align}

In the Gibbs sampler, we need to use the conditional posterior distributions of $\boldmu$, $\Sigma$ and $\Sigma_p$. It is easy to note that
\begin{equation}
\label{eq:Sigma-posterior-density}
f(\Sigma \given \boldmu, \Sigma_p, \F_j) \propto \abs{\Sigma}^{-(\nu+j+C)/2}\exp \left[-\frac{1}{2}\mbox{tr}\left(\Psi \Sigma^{-1} + \sum_{i=1}^j (n_i+0.5) (\L_i -\boldmu)(\L_i -\boldmu)^T \Sigma^{-1} \right)\right].
\end{equation}

Thus, we can write the following.
\begin{equation}
\label{eq:Sigma-posterior}
\Sigma \given \boldmu, \Sigma_p, \F_j \sim \mbox{Inverse Wishart}\left( \Psi + \sum_{i=1}^j (n_i+0.5) (\L_i - \boldmu)(\L_i - \boldmu)^T, \nu+j \right).
\end{equation}

In an identical fashion, it is possible to show that
\begin{equation}
\label{eq:Sigmap-posterior}
\Sigma_p \given \boldmu, \Sigma, \F_j \sim \mbox{Inverse Wishart}\left(\Psi_p+  ({\boldmu} -{\boldalpha}) ({\boldmu} -{\boldalpha})^T, \nu_p+1 \right).
\end{equation}

For the conditional posterior distribution of $\boldmu$, observe that
\begin{align}
& f(\boldmu \given \Sigma, \Sigma_p, \F_j) \nonumber \\ 
&\propto \exp \left[-\frac{1}{2}\left(\sum_{i=1}^j (n_i+0.5) (\L_i -\boldmu)^T \Sigma^{-1}(\L_i -\boldmu) + (\boldmu -\boldalpha)^T \Sigma_p^{-1}  (\boldmu -\boldalpha)\right)\right] \nonumber \\
\label{eq:mu-posterior-density}
&\propto \exp \left[-\frac{1}{2}\left(\boldmu^T\left[\Sigma_p^{-1} + (N_j+j/2)\Sigma^{-1} \right]\boldmu - \boldmu^T \left[\Sigma_p^{-1} {\boldalpha} + \Sigma^{-1}\left(\sum_{i=1}^j (n_i+0.5) \L_i \right) \right]\right)\right].
\end{align}

Let $V_j=\Sigma_p^{-1} + (N_j+j/2)\Sigma^{-1}$. Comparing the above expression with the density of multivariate normal distribution, and readjusting the terms as necessary, we can show that
\begin{equation}
\label{eq:mu-posterior}
\boldmu \given \Sigma, \Sigma_p,\F_j \sim \N_{C-1} \left( V_j^{-1} \left[\Sigma_p^{-1} {\boldalpha} + \Sigma^{-1}\left(\sum_{i=1}^j (n_i+0.5) \L_i \right) \right], V_j^{-1} \right).
\end{equation}

Existing theory on Bayesian inference ensures that if the above conditional posterior distributions are iterated many times, it would converge to the true posterior distributions for the parameters. It is naturally of prime importance to make sure that convergence is achieved when we implement the algorithm and collect a posterior sample. In that regard, we use the Gelman-Rubin statistic, cf. \cite{gelman1992inference}. This is an efficient way to monitor the convergence of the Markov chains. Here, multiple parallel chains are initiated from different starting values. Following the standard theory of convergence, all of these chains eventually converge to the true posterior distributions and hence, after enough iterations, it should be impossible to distinguish between different chains. Leveraging this idea and applying an ANOVA-like technique, the Gelman-Rubin statistic compares the variation between the chains to the variation within the chains. Ideally, this statistic converges to 1, and a value close to 1 indicates convergence. In all our applications, we start multiple chains by randomly generating initial values for the parameters and start collecting samples from the posterior distributions only when the value of the Gelman-Rubin statistic is below 1.1. Also, we take samples from iterations sufficiently apart from each other so as to ensure independence of the realizations. 

\subsection{Prediction}\label{subsec:prediction}

In this section, we discuss the prediction procedure for the proposed model. Note that, because of the conjugacy we observed above, the posterior predictive distribution for the $l^{th}$ instance based on the data up to the $j^{th}$ instance ($l>j$) is given by a normal distribution. The mean parameter of this distribution can be computed by the following equation: 
\begin{equation}
\label{eq:predict-posterior-mean}
\E(\L_{l} \given \F_j) = \E\left(\E\left(\L_{l} \given \boldmu,\F_j\right)\right) = \E(\boldmu\given\F_j).
\end{equation}

Similarly, using the property of conditional variance, the dispersion parameter of the posterior predictive distribution is
\begin{equation}
\label{eq:predict-posterior-var}
\cov(\L_{l} \given \F_j) = \E\left[\cov(\L_{l} \given \boldmu,\Sigma,\F_j)\right] + \cov\left(\E(\L_{l} \given \boldmu,\Sigma,\F_j)\right) = \frac{\E\left(\Sigma \given \F_j\right)}{n_l+0.5} + \cov(\boldmu\given \F_j).
\end{equation}

Once again, it is complicated to get closed form expressions for the above expectation and variance. Therefore, we make use of the Gibbs sampler to simulate realizations from the posterior predictive distribution. Based on the data up to the $j^{th}$ instance, we generate $M$ samples (for large $M$) for the parameters from the posterior distributions. Let us call them $\Sigma_{j,s}, \boldmu_{j, s}$ for $s=1,2, \ldots, M$. Then, we can approximate the mean of the posterior predictive distribution in \cref{eq:predict-posterior-mean} by $\sum_{s=1}^M\boldmu_{j, s}/M$. Similarly, the first term in the right hand side of \cref{eq:predict-posterior-var} can be approximated by $\hat \Sigma_j/(n_l+0.5)$ where $\hat\Sigma_j$ is the sample mean of $\Sigma_{j,s}$ for $s=1,2, \ldots, M$, while the second term in that equation can be approximated by the sample dispersion matrix of $\boldmu_{j, s}$ for $s=1,2,\ldots,M$. Next, we generate many samples (once again, say $M$) from the posterior predictive distribution with these estimated parameters.

This sample from the posterior predictive distribution is then used to infer on various properties of $(h(\Y_l)\mid \F_j)$, as discussed in \Cref{subsec:preliminaries}. 

\section{Simulation Study}
\label{sec:simulation} 

In this section, we consider a few toy election scenarios and evaluate the effectiveness of our proposed methodology across various scenarios of elections. We generate $K$ batches of election data as multinomial random draws following \cref{eq:proposed_model} for three different simulated election outcomes (SEO). Each election data is assumed to have $C$ candidates. The expected values of $n_j$ (the number of votes cast in the $j^{th}$ batch ) are considered to be the same for all $j$. Let $n$ denote the common expected value for different batches. We consider different value of  $C$ (3 or 5), $K$ (25 or 50), $n$ (100 or 1000 or 5000 or 50000) to understand the robustness of the method. 

The SEOs we consider differ in the way the $\boldeps_j$'s are generated. First, we take the simplest situation of $\varepsilon_{cj}=0$ for all $c,j$, which corresponds to an iid collection of multinomial random variables. Second, $\boldeps_{j}$ is simulated from a truncated multivariate normal distribution with mean $\bm{0}$ and dispersion matrix $n_j^{-0.5}\I$. Third, we impose dependence across different coordinates of $\boldeps_{j}$, and simulate it from a truncated multivariate normal distribution with mean $\bm{0}$ and dispersion matrix $n_j^{-0.5}A$, where $A$ is a randomly generated positive definite matrix with finite norm. Observe that all  the three SEOs satisfy \Cref{asmp:exp_var_epsilon}. For every combination of $K,C,n$, each experiment is repeated many times and we compute the average performance over these repetitions. 

In this simulation study, we particularly explore two aspects in detail. First, we find out the accuracy of our methodology in terms of predicting the candidate with maximum votes at the end. Here, a decision is made when the winner is predicted with at least 99.5\% probability and when the predicted difference in votes for the top two candidates is more than 5\% of the remaining votes to be counted. Second, we focus on two specific cases which are similar to our real-life applications, and examine how well we can predict various properties of $h(\Y_l)$, for different choices of $h(\cdot)$. One of these cases has $C=5, K=50$ and we look for different options for $n$. In the other case we have $C=3, K=25, n=5000$. 

In case of the first problem, for every iteration, the fixed values of $\boldp$ are obtained from a Dirichlet distribution with equal parameters to ensure a more general view of the performance of the proposed method. Detailed results are reported in \Cref{sec:additional-results}, in Tables \ref{tab:simulation-accuracy-margin-dgp1-detail}, \ref{tab:simulation-accuracy-margin-dgp2-detail}, and \ref{tab:simulation-accuracy-margin-dgp3-detail}. We note that the results are not too sensitive with respect to the choices of $C$ and $K$, but the accuracy increases with $n$. Below, in \Cref{tab:simulation-accuracy-margin}, we look at the results for $C=5$, $K=50$ for the three different SEOs. It is evident that under the iid assumption (SEO1), the method makes a correct call almost all the times. When the components of $\boldeps_j$ are independent random variables (SEO2), a correct call is made more than 90\% of the times for large enough samples ($n>5000$). For the third SEO, where the components of $\boldeps_j$ are simulated from a multivariate distribution with a general positive definite matrix as the covariance matrix, the accuracy of making a correct call drops to 70\% for $n=50000$. Another interesting observation is that the accuracy depends heavily on the average final margin between the top two categories. In other words, if the underlying probabilities of the top two candidates are close (which directly relates to a lower margin of difference between the final counts), then the method achieves lower accuracy, and vice-versa. We also observe that in those cases, the method records {\it no call} (i.e.\ a decision cannot be made with the prescribed rules) more often than incorrect calls.
\begin{table}
\caption{\label{tab:simulation-accuracy-margin} Accuracy (in \%) of predicting the category with maximum votes (corresponding final margins, averaged over all repetitions, are given in parentheses) for different SEOs. All results correspond to the case of $C=5$, $K=50$.}
\centering
\begin{tabular}{|lllll|}
	\hline
	SEO & $n$   & Correct (avg margin) & Incorrect (avg margin) & No call (avg margin) \\
	\hline
	SEO1 & 100 & 94.5\% (1584) & 1.5\% (73) & 4\% (168) \\ 
	& 1000 & 98.25\% (14388) & 0.5\% (516) & 1.25\% (378) \\ 
	& 5000 & 99.5\% (71562) & 0.25\% (1461) & 0.25\% (858) \\ 
	& 50000 & 100\% (730428) &  &  \\
	\hline
	SEO2 & 100 & 70\% (605) & 8.25\% (205) & 21.75\% (314) \\ 
	& 1000 & 84.75\% (10986) & 4.75\% (1286) & 10.5\% (2891) \\ 
	& 5000 & 91.75\% (67476) & 3.25\% (6624) & 5\% (12241) \\ 
	& 50000 & 95.5\% (720905) & 2.25\% (71467) & 2.25\% (77728) \\
	\hline
	SEO3 & 100 & 42.4\% (154) & 22\% (106) & 35.6\% (110) \\ 
	& 1000 & 45\% (2439) & 23.2\% (1423) & 31.8\% (1536) \\ 
	& 5000 & 55.8\% (18972) & 14.2\% (9671) & 30\% (10822) \\
	& 50000 & 70\% (305700) & 8.8\% (79016) & 21.2\% (184868) \\ 
	\hline
\end{tabular}
\end{table}

Moving on to the particular choice of $C=3,K=25,n=5000$, we aim to find out the effectiveness of the proposed approach in identifying the leading candidates for varying degrees of difference between the true vote probabilities of the topmost two candidates. To that end, for $\boldp=(p_1,p_2,p_3)$, without loss of generality, we assume $p_1>p_2>p_3$. Now, data are generated for the above three simulated elections by fixing $p_1-p_2=\delta$ where $\delta\in \{0.01,0.05,0.1,0.25\}$. For every choice of $\delta$ and for every SEO, we repeat the experiments and find out the mean accuracy in predicting the top candidate. Additionally, we also find out the average number of observations used by the method for making the prediction. These results are displayed in \Cref{tab:simulation-accuracy-vote}.

\begin{table}
\caption{\label{tab:simulation-accuracy-vote} Accuracy (in \%) of predicting the candidate with maximum vote count for different SEOs and for different values of $\delta =p_1-p_2$. All results correspond to the case of $n=5000$, $C=3$, $K=25$. Numbers inside the parentheses indicate what percentage of data are used on an average before making a call.}
\centering
\begin{tabular}{|lllll|}
	\hline
	SEO & $\delta$   & Correct (data used) & Incorrect (data used) & No call \\
	\hline
	DGP1 & 0.01 & 96\% (14\%) & 4\% (13\%) &   \\ 
	& 0.05 & 100\% (12\%) &  &  \\ 
	& 0.10 & 100\% (12\%) &  &  \\  
	& 0.25 & 100\% (12\%) &  &  \\
	\hline
	SEO2 & 0.01 & 64.5\% (22\%) & 30.5\% (14\%) & 5\%  \\ 
	& 0.05 & 76\% (22\%) & 22.5\% (14\%) & 1.5\%  \\ 
	& 0.10 & 89.5\% (16\%) & 10\% (14\%) & 0.5\%  \\ 
	& 0.25 & 100\% (12\%) &  &  \\
	\hline
	SEO3 & 0.01 & 48\% (26\%) & 34\% (14\%) & 18\%  \\ 
	& 0.05 & 63.5\% (24\%) & 21\% (13\%) & 15.5\%  \\ 
	& 0.10 & 68\% (22\%) & 16.5\% (15\%) & 15.5\%  \\ 
	& 0.25 & 74.5\% (16\%) & 8.5\% (13\%) & 17\%  \\ 
	\hline
\end{tabular}
\end{table}

It is evident that for SEO1, the method provides great prediction. Even when there is only 1\% difference in the probabilities of the top two candidates, the method predicts the leading candidate correctly 96\% of the times. For larger than 1\% difference, it never fails to predict the leading candidate correctly. For SEO2, our method can predict the topmost candidate with more than 75\% accuracy whenever the probabilities of the top two candidates differ by at least 5\%. This accuracy reaches the value of nearly 90\% for $\delta = 0.1$ and is 100\% if the difference in those two cell probabilities is 0.25. For the third SEO, the prediction accuracy is about 70\% for $\delta=0.1$. The accuracy improves steadily as $\delta$ increases. One can say that the method will be able to predict the winner with high level of accuracy whenever there is a considerable difference between the $p_i$ values for the top two candidates. If that difference is minute, which corresponds to a very closely contested election, the accuracy will drop. Additionally, we observe that in the most extreme cases, about 20\% of the times, the method never declares a winner with desired certainty, which is in line with what we should expect.

We now focus on the second specific case of $C=5,K=50,n=50000$. For different SEOs, the values of $\boldp$ are generated randomly from a Dirichlet distribution, and we evaluate how well our method can predict the cumulative proportions of counts for the five categories. Experiments are repeated many times and we compute the overall root mean squared error (RMSE) for every SEO, based on the predictions made at different stages. Refer to \Cref{tab:simulation-rmse-sales} for these results. Once again, for SEO1, the method achieves great accuracy very early. With only 15 rounds of data (approximately 30\% of the total observations), the predictions are precise. At a similar situation, for SEO2, the RMSE is less than 1\%, and it decreases steadily to fall below 0.5\% by round 35. Finally, for the third SEO, RMSE of less than 1\% is recorded after 35 rounds as well. These results show that our method can accurately predict the overall cumulative proportions of counts with about 35 rounds of data. Translating this to the sales forecasting problems, we hypothesize that the proposed method will be able to predict the annual market shares of different categories correctly at least 3 to 4 months before the year-end.

\begin{table}
\caption{\label{tab:simulation-rmse-sales} Root mean squared error (RMSE, in \%) in estimating the overall cumulative proportions of counts for all categories.}
\centering
\begin{tabular}{|lccc|}
	\hline
	Data used & SEO1 & SEO2 & SEO3  \\
	\hline
	5 rounds  & 0.06\%    & 1.95\% & 4.38\% \\ 
	15 rounds & $<0.01\%$ & 0.95\% & 2.21\% \\ 
	25 rounds & $<0.01\%$ & 0.63\% & 1.46\% \\ 
	35 rounds & $<0.01\%$ & 0.41\% & 0.92\% \\ 
	45 rounds & $<0.01\%$ & 0.22\% & 0.51\% \\
	\hline
\end{tabular}
\end{table}

\section{Application on the data from Bihar election}
\label{sec:Bihar_election}

We examine the effectiveness of our proposed methodology in the election calling context, with the real data from  the Bihar Legislative Assembly Election held in 2020, across 243 constituencies (seats). Indian elections have multi-party system, with a few coalition of parties formed before (sometimes after) the election, playing the major role in most contests. While these coalitions are at times temporary in nature and there is no legal binding or legitimacy, typically the alliance winning  the majority (at least 50\%) of the constituencies comes to power. In the Bihar 2020 election, the contests in most constituencies were largely limited to two such alliances --  the National Democratic Alliance (NDA), and the {\it Mahagathbandhan} (MGB). Although there were a couple of other alliances, they did not impact the election outcomes and for the purpose of this analysis, we consider them in ``others'' category. For each constituency, the counting of votes took place in several rounds, and we want to examine how these round-wise numbers can be useful in predicting the final winner ahead of time.

The result of the election is downloaded from the official website of \cite{bihardata}. A brief summary of the data is provided in \Cref{tab:bihar-summary}. All summary statistics in this table are calculated over the 243 constituencies in Bihar (our analysis ignores the postal votes, which are few in numbers and do not affect the results). Note that the number of votes cast in the constituencies vary between 119159 and 225767, with the average being 172480. The votes are counted in 25 to 51 rounds across the state, and the mode is found out to be 32. While in most rounds about little over 5000 votes are counted, in some cases, especially the last few rounds, show fairly less numbers. The last three rows in the table provide the minimum, maximum and average number of votes counted per round for all the constituencies. We also point out that the final margin between winner and runner-up parties in these constituencies show a wide range. The lowest is recorded in Hilsa (13 votes) whereas the maximum is recorded in Balrampur (53078 votes).

\begin{table}[!hbt]
\caption{\label{tab:bihar-summary} Summary of the Bihar election data. All summary statistics are calculated over the 243 constituencies.}
\centering
\begin{tabular}{|lcccc|}
	\hline
	& Minimum & Maximum & Average & Median \\
	\hline
	Total votes cast & 119159 & 225767 & 172480 & 172322 \\
	Number of rounds for counting & 25 & 51 & 32.47 & 32 \\
	Final margin & 13 & 53078 & 16787 & 13913 \\
	Minimum votes counted per round & 114 & 5173 & 1618 & 1029 \\
	Maximum votes counted per round & 4426 & 8756 & 6745 & 6728 \\
	Average votes counted per round & 3070 & 6505 & 5345 & 5381 \\
	\hline
\end{tabular}
\end{table}

As mentioned above, in the following analysis, we focus on the two dominant categories (essentially the two main alliances in every constituency) and the rest of the alliances or the parties are collated into the third category. As our main objective is to predict the winner and the margin of victory, it is sensible to consider this three-category-setup. We also make the logical assumption that each vote is cast independently, thereby ensuring that in each round of counting, we have an independent multinomial distribution. However, there is no information on how the counting happens sequentially in each round, and therefore we assume the individual probabilities of the three different categories in different counting rounds to be random variables satisfying \Cref{asmp:exp_var_epsilon}. Consequently, it is an exciting application of the proposed modeling framework.

Recall the hierarchical prior distributions for the parameters in the model (\cref{eq:hierarchical-priors}). For $\boldalpha$, we use the past election's data and use the proportion of votes received by the corresponding alliances (scaled, if needed, to have $\sum \alpha_i=1$). For $\Psi$ and $\Psi_p$, we use identity matrices of appropriate order. Both $\nu$ and $\nu_p$ are taken to be 5. We point out that the results are not too sensitive to these choices. Next, using these priors, we implement our method and estimate the probability of winning for every alliance after the counting of votes in each round. Corresponding margin of victory is also estimated in the process. Based on these values, we propose the following rules to call a race. First, we allow at least 50\% votes to be counted in order to make a prediction. It is otherwise considered to be {\it too early to call}. On the other hand, akin to the procedure laid out by \cite{NBCnews}, our decision is to call the race in a particular constituency when we are at least 99.5\% confident of the winner and when the predicted winning margin is at least 5\% of the remaining votes. Unless these conditions are met, it is termed as {\it too close to call}.

We start with a summary of the results obtained after fitting the model to the election data. Complete results are provided in \Cref{tab:bihar-fullresults} in \Cref{sec:additional-results}. In 227 out of the 243 constituencies (approximately 93.4\%), the proposed method calls the race correctly. Our method considers that it was always {\it too close to call} for Bakhri and Barbigha where the eventual win margin was 439 and 238 respectively. For 14 other constituencies (approximately 5.7\%), the method calls the race incorrectly, as the winning candidate was significantly behind at the time of calling in those constituencies.

We next take a detailed look at the prediction patterns for all constituencies, in different aspects, through \Cref{fig:bihar-graph1}. In both panels of the Figure we color the correct predictions by green, the incorrect predictions by red, and no calls by blue. In the left panel of the Figure we plot the remaining vote percentages at the time of calling along the color gradient. In the right panel we plot the final margins along the color gradient. It is interesting to observe that in the right panel all red points, which correspond to incorrect calls, appear with lighter shades. Thus, our method performs perfectly for all constituencies but one where the final margins of victories are above 10000. From the left panel it is also evident that the proposed method calls the races very early for all such constituencies. Only exception to this is the constituency Baisi, which is discussed in more detail below. We can also observe that seats painted in lighter shades in left panel are also in lighter shades on right panel, implying that the closely contested constituencies require more votes to be counted before the race can be called.

\begin{figure}[!hbt]
\centering
\includegraphics[width = \textwidth,keepaspectratio]{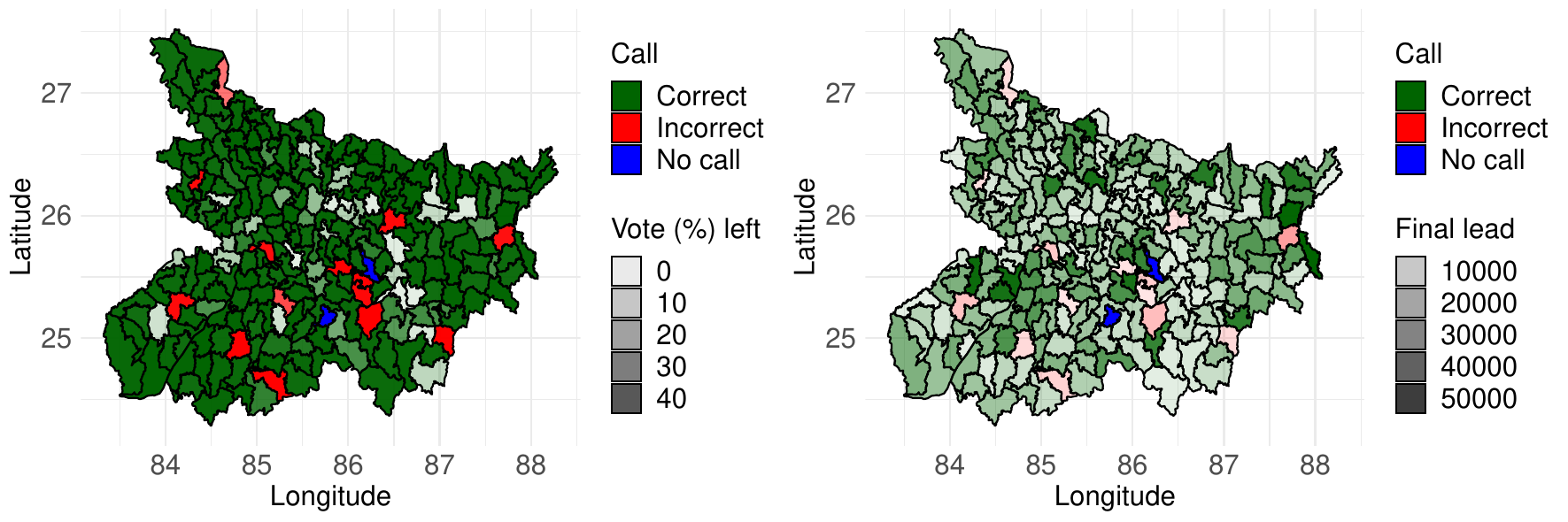}
\caption{Forecasting accuracy for different constituencies in the Bihar election data. Different colours indicate whether it is a correct call or not. Left panel shows how quickly a call is made, as a darker color indicates that the call is made with higher proportion of votes remaining to be counted. In the right panel, a lighter shade indicates smaller final margin, i.e.\ a closely contested election.}
\label{fig:bihar-graph1}
\end{figure}



To further explore the above point, the accuracy and the average percentage of votes counted before making a call, corresponding to the final margins of victories, are displayed in \Cref{tab:bihar-correct-call-summary}. There are 52 constituencies which observed very closely fought election (final margin is less than 5000 votes). In 41 out of them, our method correctly predicts the winner and in 2 of them, it cannot make a call. On average, around 70\% votes are counted for all these constituencies before we can make a call. This, interestingly, drops drastically for the other constituencies. On average, only about 60\% votes are counted before we make a call in the constituencies where eventual margin is between 5000 and 10000. In 4 out of 34 such constituencies though, our method predicts inaccurately. For the constituencies with higher eventual margins (greater than 10000),  the prediction turns out to be correct in 156 out of 157 constituencies. In these cases, around only about 50 to 55 percentage of votes are counted on average before making a call. Finally, the last row of the table corroborates the earlier observation that the size of the constituency does not have considerable effect on the prediction accuracy.

\begin{table}[h]
\caption{\label{tab:bihar-correct-call-summary} Summary of the prediction accuracy in terms of calling the race in different constituencies, according to the final margin of victory.}
\centering
\begin{tabular}{|lcccccc|}
	\hline
	Final margin & $<2000$ & 2000--5000 & 5000--10000 & 10000--20000 & $>20000$ & Total \\
	\hline
	Constituencies & 23 & 29 & 34 & 74 & 83 & 243 \\
	Correct call & 16 & 25 & 30 & 73 & 83 & 227 \\
	Incorrect call & 5 & 4 & 4 & 1 & 0 & 14 \\
	Too close to call & 2 & 0 & 0 & 0 & 0 & 2 \\
	Average counting & 68.0\% & 74.1\% & 60.5\% & 55.2\% & 51.9\% & 58.3\% \\
	Average votes & 172136 & 170946 & 171107 & 171636 & 174427 & 172480 \\
	\hline
\end{tabular}
\end{table}

Not only the winner of an election, but the proposed approach also predicts the final margin of victory, along with a prediction interval, for the winner. In \Cref{fig:bihar-prediction-interval}, we present the true margin and the predicted win margin for all the 227 constituencies where correct calls are made. Corresponding prediction intervals are also shown in the same plot. It can be observed that the predicted margins closely resemble the true margins.

\begin{figure}[!hbt]
\centering
\includegraphics[width = 0.8\textwidth,keepaspectratio]{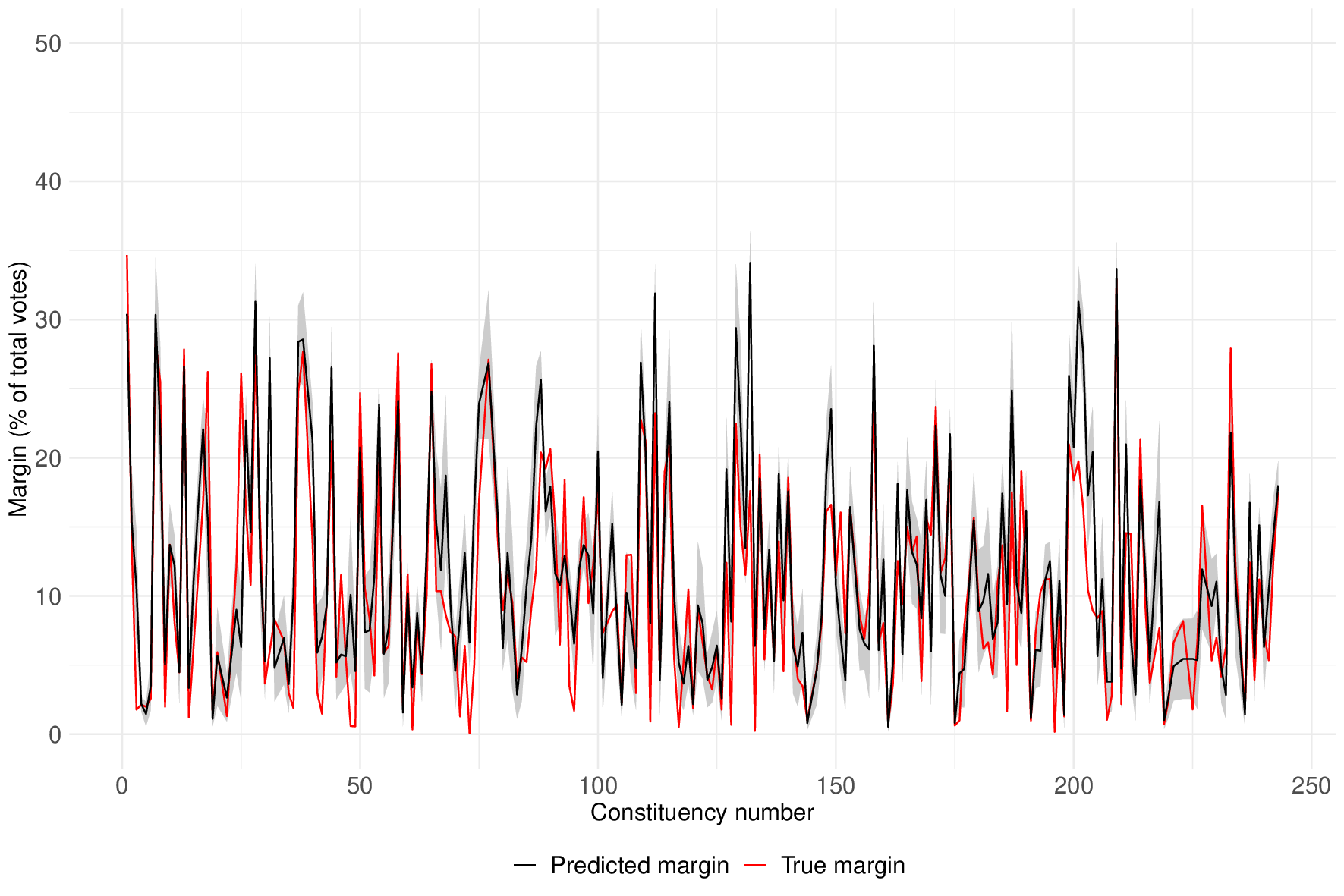}
\caption{True margin of victory and the predicted margin of victory (both in \% of total votes) for all constituencies where the method made a correct call. Prediction interval is displayed in grey.}
\label{fig:bihar-prediction-interval}
\end{figure}

Overall, it is evident that the method performs really well in terms of calling an election race based on only the counting of votes in different rounds. It is imperative to point out that better accuracy can be achieved with more detailed information about different constituencies. Especially, information about the sequence in which the counting is carried out, ethnic composition or socio-economic status of people in different areas are necessary to build a more accurate procedure. Our proposed method in this paper works with only the multinomial assumption and achieves more than 93\% accuracy in calling the race with the requirement of at least 50\% votes being counted. To that end, in order to gain more insight about the accuracy of the method, we rerun the experiment with different values for the minimum requirement of counting. The number of constituencies with correct call, incorrect call and no call are displayed against the minimum requirement in \Cref{tab:bihar-minimum-requirement}. We see that the method achieves 90\% accuracy (215 out of 243) even when only 30\% votes are counted. For minimum counting of 60\% or more votes, the results do not improve substantially, except for the fact that no call can be made in more number of constituencies. 

\begin{table}[ht]
\centering
\caption{Prediction accuracy across all 243 constituencies for the proposed method corresponding to minimum \% of votes required to be counted before calling a race.}
\label{tab:bihar-minimum-requirement}
\begin{tabular}{cccc}
	\hline
	Minimum requirement & Correct call & Incorrect call & Too close to call \\ 
	\hline
	10\% & 182 &  61 &   0 \\ 
	20\% & 199 &  44 &   0 \\ 
	30\% & 215 &  28 &   0 \\ 
	40\% & 221 &  21 &   1 \\ 
	50\% & 227 &  14 &   2 \\ 
	60\% & 234 &   7 &   2 \\ 
	70\% & 235 &   4 &   4 \\ 
	80\% & 235 &   3 &   5 \\ 
	90\% & 235 &   1 &   7 \\ 
	\hline
\end{tabular}
\end{table}


Next, we look into the robustness of the method, by focusing on four particular constituencies in more detail. These are Hilsa (minimum final margin), Balrampur (maximum final margin), Baisi (final margin more than 10000, but our proposed methodology makes a wrong call with nearly 55\% votes counted) and Barbigha (final margin of 238, the proposed method never calls the race in anyone's favour). Now, for each constituency, we create synthetic data by randomly switching the order of the individual rounds of counting and implement our method on the resulting data. This exercise helps us to understand whether the predictions would have been different if the data were available in a different order. In particular, it points towards objectively judging to what extent the success or failure of the proposed methodology in this real data can be attributed to luck. Note that the proposed modeling framework assumes independence across different rounds, and hence the method should perform similarly if the votes are actually counted in random order. In \Cref{tab:bihar-permutation} below, the results of these experiments are presented. We see that for Balrampur, even with randomly permuted rounds, the method calls the race correctly every time. For Baisi, the method yields correct forecast 99.97\% of times, although in the original order of the data it makes an incorrect call. Meanwhile, Hilsa and Barbigha observed very closely contested elections. For these two constituencies, in the experiment with permuted data, the method records {\it too close to call} decisions more commonly, 72.55\% and 70.53\% respectively. These results suggest that quite possibly the sequence of roundwise counting of  votes is  not random. 
Thus, if more information is available about individual rounds from which votes are counted, it would be possible to modify the current method to achieve greater accuracy.

\begin{table}
\caption{\label{tab:bihar-permutation} Summary of the prediction accuracy in four constituencies, for synthetic data generated through many permutations of individual rounds.}
\centering
\begin{tabular}{|lccccc|}
	\hline
	& & & \multicolumn{3}{c|}{Results in permuted data} \\
	Constituency & Final margin & Original call & Correct & Incorrect & No call \\
	\hline
	Hilsa & 13 & Incorrect & 12.57\% & 14.87\% & 72.55\% \\
	Balrampur & 53078 & Correct & 100\% & 0\% & 0\% \\
	Baisi & 16312 & Incorrect & 99.97\% & 0.03\% & 0\% \\
	Barbigha & 238 & No call & 17.19\% & 12.28\% & 70.53\% \\
	\hline
\end{tabular}
\end{table}

\section{Application to the USA election with hypothetical sequences}
\label{sec:usa_application}

In this section, in an attempt to impress upon the reader that the proposed methodology can be extended to other election forecasting contexts as well, we focus on the statewise results from the 2020 Presidential Election of the USA (data source: \cite{usdatarepo}). It is well known that in American politics, every election relies primarily on a few states, popularly termed as the `swing states', where the verdict can reasonably go to either way. For the other states, the results are guessable with absolute certainty, and are of less interest in any such research. Thus, in this application, our focus is on the eleven swing states of the 2020 Presidential Election -- Arizona, Florida, Georgia, Michigan, Minnesota, Nevada, New Hampshire, North Carolina, Pennsylvania, Texas and Wisconsin. Since the final outcome of the election in these cases are extremely likely to be either democratic or republican, with a minimal chance of a third-party-victory, we can consider a three-category-setup. Then, akin to \Cref{sec:Bihar_election}, it can be easily argue that the setup follows the proposed modeling framework. We also adopt the same setting for prior distributions and the procedure to call the race. 

Our objective is to demonstrate that partial data of the counties can provide accurate forecast of the final result in these turbulent states. However, we note that the information of the particular order in which results were declared for different counties are unavailable. That is why we call it a hypothetical application, and we circumvent the issue by taking a permutation-based approach to implement the proposed method. For every state, we randomly assign the orders to the counties and assume that the data are available in that sequence. This permutation procedure is repeated 100 times, and we evaluate the performance of the proposed method at a summary level across these permutations. 

Let us start with a brief overview of the number of counties and the number of voters in the swing states (see \Cref{tab:usa-summary}). All summary statistics are calculated based on the eleven states. We can see that the number of counties as well as the number of voters vary substantially; although the final margin of victory remains within 7.4\% of the total votes. It is in fact less than 1\% in three states (Arizona, Georgia and Wisconsin), making them the most closely contested states in this election.

\begin{table}[!hbt]
\caption{\label{tab:usa-summary} Summary of the 2020 US Presidential Election data. All summary statistics are calculated based on the 11 swing states.}
\centering
\begin{tabular}{|lcccc|}
\hline
& Minimum & Maximum & Average & Median \\
\hline
Number of counties & 10 & 254 & 84.64 & 72 \\
Total votes cast & 803,833 & 11,315,056 & 5,229,904 & 4,998,482 \\
Final margin (of total votes) & 0.25\% & 7.37\% & 2.94\% & 2.40\% \\
Minimum votes per county & 66 & 16515 & 3136 & 1817 \\
Maximum votes per county & 4426 & 2,068,144 & 914,028 & 755,969   \\
Average votes per county & 31427 & 225,686 & 85322 & 66739 \\
\hline
\end{tabular}
\end{table}

Turn attention to the results of the forecasting procedure on the permuted datasets. In the last three columns of \Cref{tab:us-call-summary}, we display the proportion of times our method correctly calls the race, finds it too close to call, or incorrectly calls the race based on partial county-level data, 

\begin{table}[ht]
\centering
\caption{Performance summary for the application on US Presidential Election data. Total number of votes is given in millions, final margin is given as percentage of total votes, and the last three columns indicate percentage of times the algorithm made a correct call, could not call or made a wrong call.}
\label{tab:us-call-summary}
\begin{tabular}{|lcccccc|}
\hline
State & Counties & Total votes & Final margin & Correct & Too close & Incorrect \\ 
\hline
Arizona &  15 &  3,385,294 &  10,457 & 60\% & 39\% & 1\% \\ 
Florida &  67 & 11,067,456 & 371,686 & 99\% & 0\% & 1\% \\ 
Georgia & 159 &  4,998,482 &  12,670 & 59\% & 1\% & 40\% \\ 
Michigan &  83 &  5,539,302 & 154,188 & 72\% & 0\% & 28\% \\ 
Minnesota &  87 &  3,277,171 & 233,012 & 82\% & 0\% & 18\% \\ 
Nevada &  17 &  1,404,911 &  33,706 & 97\% & 3\% & 0\% \\ 
New Hampshire &  10 &    803,833 &  59,277 & 100\% & 0\% & 0\% \\ 
North Carolina & 100 &  5,524,802 &  74,481 & 74\% & 1\% & 25\% \\ 
Pennsylvania &  67 &  6,915,283 &  80,555 & 54\% & 4\% & 42\% \\ 
Texas & 254 & 11,315,056 & 631,221 & 96\% & 0\% & 4\% \\ 
Wisconsin &  72 &  3,297,352 &  20,608 & 54\% & 1\% & 45\% \\
\hline
\end{tabular}
\end{table}

We observe that in New Hampshire, Florida, Nevada and Texas, partial data provide the correct prediction in almost all cases. One may connect this to the fact that the final margin in these states were on the higher side, especially if we look at the values in terms of proportion to the total number of votes. Arizona, on the other hand, was the most closely contested state with a big voter pool. Our approach is still able to predict the correct winner 60\% of the times whereas it decides the race to be too close to call on 39\% occasions. Georgia and Wisconsin are the two other states with the lowest final margins. There, our algorithm makes wrong judgements in more than 40\% situations. 

In the other four states, namely Michigan, Minnesota, North Carolina and Pennsylvania, the results are not good. Albeit the final margins are considerably large, the final conclusions are wrong in at least 18\% of the permuted datasets. This phenomena brings forward an important limitation of the methodology. In our procedure, the key assumption is that the perturbations, albeit random and different in different counties, have the same asymptotic distribution. In such political applications, this assumption may be violated on certain occasions. To demonstrate it empirically, let us look at the margins of differences across the counties in these states (\Cref{fig:usa-plot}).

\begin{figure}[!hbt]
\centering
\includegraphics[width = \textwidth,keepaspectratio]{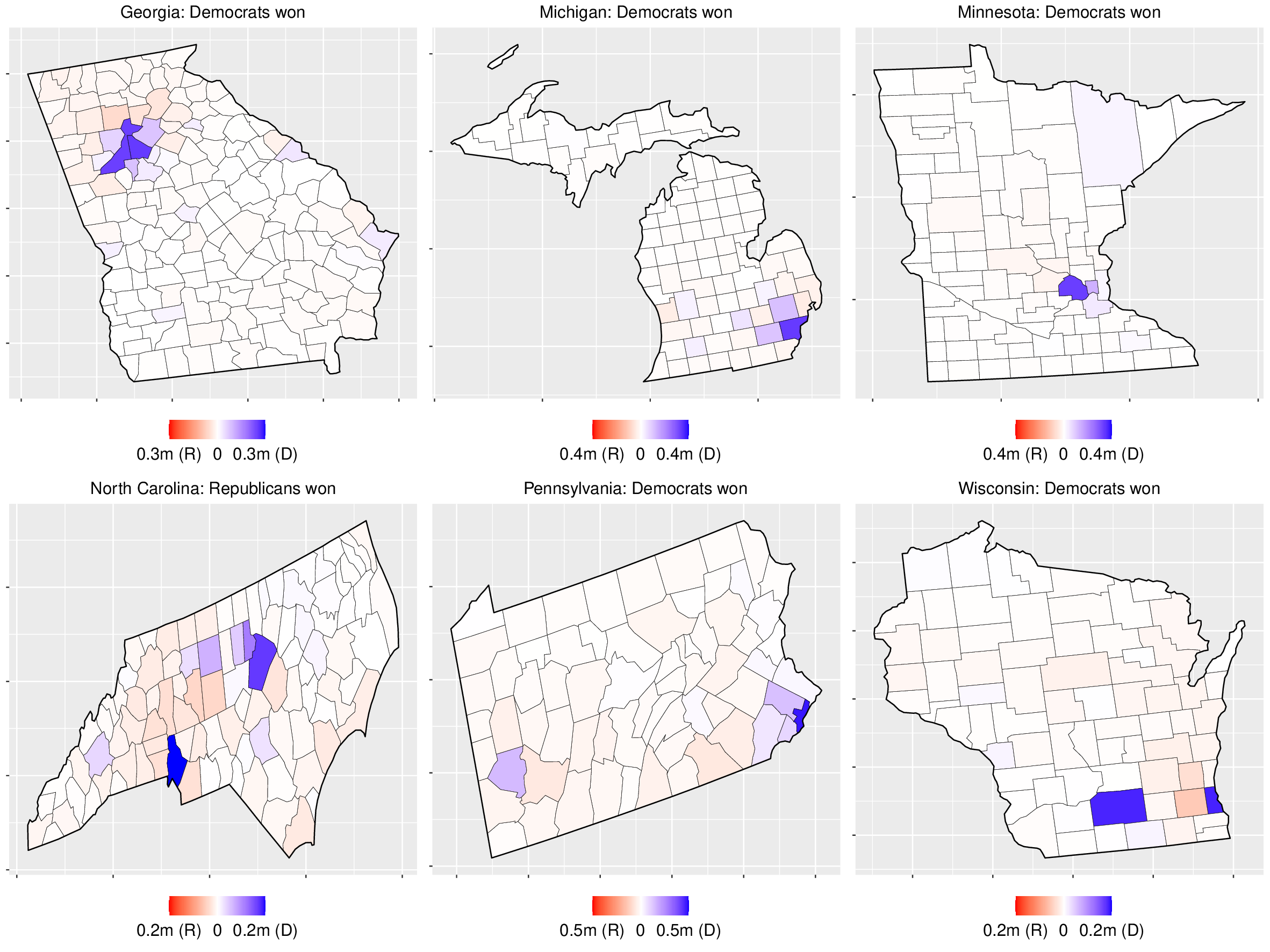}
\caption{County-wise margin of difference in votes for the two main parties (D and R stand for democratic and republican party respectively) in the six states where the method did not perform well. The unit ``m'' indicates millions.}
\label{fig:usa-plot}
\end{figure}

From the figure, it is evident that in all these states (most prominently in Michigan, Minnesota, Georgia and Wisconsin), the final verdict relied heavily on only a couple of counties. Quite naturally, a partial dataset with the information of those particular counties is extremely likely to provide contradictory results to a partial dataset without those counties. In North Carolina, on the other hand, several counties observed high margin of difference in favor of the democratic party, while the final ruling went the opposite way. These situations point to a potential deviation from the aforementioned assumption, and we hypothesize that in such circumstances, the method can be updated to include more information about individual counties. It calls for an interesting future research direction and we discuss it in more detail in the next section.

\section{Conclusions}
\label{sec:conclusion}

In this work, we consider a hierarchical Bayes model for batches of multinomial data with cell probabilities changing randomly across batches. The goal is to predict various properties for aggregate data correctly and as early as possible. We illustrate the effectiveness of the methodology in poll prediction problems which motivated this study. The performance of the methodology is good, especially considering the limited information. We plan to implement the method in future elections. A potential future direction in this specific context would be to improve the method when more information on the constituencies or the rounds of counting are available. It is likely that additional information on relevant covariates would improve forecast accuracy. It is possible to modify the proposed methodology to incorporate the covariates, but that is a more challenging and interesting problem and requires full treatment. We defer it for a future work. 

The proposed model accommodates randomness in the multinomial cell probabilities, depending on the batch. This is pragmatic given the actual data pattern in most practical situations, as otherwise early calls are made which end up being often wrong. This is also intuitively justifiable as, for example, different rounds of votes can have differential probabilities for the candidates, potentially owing to different locations or the changes in people's behaviour.  This paper aims to tackle these situations appropriately.

It is worth mention that a common practice in Bayesian analysis of sequences of multinomial data is to apply Dirichlet priors for the cell probabilities. This approach requires the iid assumption and hence do not work under the modeling framework of this study. In fact, if we assume iid behaviour for the multinomial data across the batches and use that model for the examples discussed above, then the performances are much worse than what we achieve. Further, we want to point out that even without the variance stabilizing transformation, the proposed Bayesian method in conjunction with the results from \Cref{thm:unconditional-x-y} can be used in similar problems. The prediction accuracy of that model is comparable to our approach. We employ the transformation for superior prediction performance, albeit marginally in some cases, across all scenarios.

Finally, note that in some contexts, it may be unrealistic to assume that values of the future $n_j$'s are known. In that case, the decision can be taken by anticipating them to be at the level of the average of $n_j$'s observed so far. It is intuitively clear from \Cref{subsec:prediction}, as one can argue that the effectiveness of the methodology will not alter to any appreciable extent as long as the sample sizes are large. It can also be backed up by relevant simulation studies.

The methodology proposed in this work can be applied in other domain as well, most notably in sales forecasting. Retail giants and companies track the sales for different products in regular basis. This is useful from various perspectives. One of the key aspects is to project or assess the overall (annual) sales pattern from the data collated at monthly, weekly or even daily level.  It is of great value to a company to infer about its eventual market share for the period (year) as early as possible, hence to validate if its (possibly new) marketing plans are as effective as targeted. The retailer would find these inferences useful from the perspective of managing their supply chain, as well as in their own marketing. Other possible applications of the proposed methodology (possibly with necessary adjustments) include analyzing hospital-visits data, in-game sports prediction etc. Some of these applications can make use of additional covariate information; we propose to take that up in a follow-up article.

\section{Proofs}
\label{sec:proofs}

\begin{proof}[\Cref{thm:unconditional-x-y}]
Following \cref{eq:proposed_model}, $\E(X_{lj}\given \varepsilon_{lj})=n_j(p_l+\varepsilon_{lj})$. Then, using \Cref{asmp:exp_var_epsilon} and applying the relationship between total and conditional expectation, 
\begin{equation}
\label{eq:proof_unconditional_exp_x}
\E(X_{lj}) = \E\left[\E(X_{lj}\given \varepsilon_{lj})\right] = n_jp_l.
\end{equation}

We also know that $\var(X_{lj}\given \varepsilon_{lj}) = n_j(p_l+\varepsilon_{lj})(1-p_l-\varepsilon_{lj})$. Then, we can use the relationship $\var(X_{lj}) = \var\left(\E(X_{lj} \given \varepsilon_j)\right) + \E\left(\var(Y_{lj} \given \varepsilon_j)\right)$ to get the following:
\begin{equation}
\label{eq:proof_unconditional_var_x}
\var(X_{lj}) = \var\left(n_j(p_l+\varepsilon_{lj})\right) + \E\left[n_j(p_l+\varepsilon_{lj})(1-p_l-\varepsilon_{lj})\right] = n_j p_l (1 - p_l) + n_j(n_j -1) \var(\varepsilon_{lj}).
\end{equation}

Finally, note that the covariance of $X_{1j}$ and $X_{2j}$, conditional on $\boldeps_{j}$, is equal to $-n_j(p_1+\varepsilon_{1j})(p_2+\varepsilon_{2j})$. One can use this and apply similar technique as before to complete the proof of part (a). 

Part (b) follows from (a) using the fact that the multinomial data in different batches are independent of each other.

To prove part (c), we use the multinomial central limit theorem which states that if $(D_1,D_2,\ldots,D_{C}) \sim \multi(n,\pi_1,\pi_2,\ldots, \pi_C)$, then 
\begin{equation}
\label{eq:multinomial_clt}
\sqrt n\left(\begin{pmatrix} D_1/n \\ D_2/n \\ \vdots \\ D_{C-1}/n \end{pmatrix} - \begin{pmatrix} \pi_1 \\ \pi_2 \\ \vdots \\ \pi_{C-1} \end{pmatrix} \right) \overset{\mathcal{L}}{\rightarrow} \mathcal{N}_2 \left(0, \begin{bmatrix} \pi_1 (1-\pi_1) & -\pi_1 \pi_2 & \hdots & -\pi_1\pi_{C-1} \\ -\pi_1 \pi_2 & \pi_2(1-\pi_2) & \hdots & -\pi_2\pi_{C-1} \\ \vdots & \vdots & \ddots & \vdots \\ -\pi_1 \pi_{C-1} & -\pi_2\pi_{C-1} & \hdots & \pi_{C-1}(1-\pi_{C-1})\end{bmatrix}\right).
\end{equation}

Let us define $V_{\boldeps_j}$ to be the covariance matrix of $(\hat p_{1j}, \hat p_{2j}, \ldots, \hat p_{(C-1)j})$, conditonal on $\boldeps_j$. Then, for the $(c,c')^{th}$ element of $V_{\boldeps_j}$, we have 
\begin{equation*}
(V_{\boldeps_j})_{c,c'} = \begin{cases} (p_c+\varepsilon_{cj}) (1-p_c-\varepsilon_{cj}) & \text{for } c=c', \\ -(p_c+\varepsilon_{cj}) (p_{c'}+\varepsilon_{c'j}) & \text{for } c=c'. \end{cases}
\end{equation*}

Now, under \Cref{asmp:exp_var_epsilon} and from \cref{eq:multinomial_clt}, as $n_j\to\infty$,
\begin{equation}\label{eq:condasymp-p}
\sqrt{n_j} V_{\boldeps_j}^{-1/2}\left(\hat \boldp_j - \tilde \boldp - \boldeps_j \right) \mid \boldeps_j \overset{\mathcal{L}}{\rightarrow} \N_{C-1} \left(\bm{0}, \I \right). 
\end{equation}

Using the assumptions in part (c) of \Cref{thm:unconditional-x-y}, as $n_j \to \infty$, $V_{\boldeps_j}$ converges in probability to $\Xi_p$. Hence, by an application of Slutsky's Theorem, \cref{eq:condasymp-p} is equivalent to: 
\begin{equation}\label{eq:condasymppred}
\sqrt{n_j}\; \Xi_p^{-1/2}\left(\hat \boldp_j - \tilde \boldp - \boldeps_j \right) \mid \boldeps_j \overset{\mathcal{L}}{\rightarrow} \N_{C-1} \left(\bm{0}, \I \right). 
\end{equation}

Let $f_{\boldeps}(\cdot)$ denote the density function of $\sqrt{n_j} \boldeps_j$, $\phi(\cdot,\Xi_{\boldeps})$ be the density function of $\mathcal{N}_{C-1}(\bm 0, \Xi_{\boldeps})$, and $\Phi(\cdot,\Xi_{\boldeps})$ be the distribution function of $\mathcal{N}_{C-1}(\bm 0, \Xi_{\boldeps})$. Note that the previous convergence in \cref{eq:condasymppred} is conditional on $\boldeps_j$, as $n_j \to \infty$. To find the unconditional distribution of $\hat \boldp_j - \tilde \boldp$, we look at the following, for $z \in \mathbb{R}^{C-1}$:  
\begin{equation}\label{eq:unconddisthatp}
\P(\sqrt{n_j} (\hat \boldp_j - \tilde \boldp) \leqslant z) = \int_{\mathbb{R}^{C-1}} \P(\sqrt{n_j}(\hat \boldp_j - \tilde \boldp - \boldeps_j) \leqslant z - y \mid \boldeps_j) f_{\boldeps}(y)\; dy. 
\end{equation}

From the assumptions in part (c) of \Cref{thm:unconditional-x-y}, $\sqrt{n_j} \boldeps_j$ converges in distribution to $\mathcal{N}_{C-1}(\bm 0, \Xi_{\boldeps})$. Then,
\begin{equation}
\lim_{n_j \to \infty} \P(\sqrt{n_j}(\hat \boldp_j - \tilde \boldp - \boldeps_j) \leqslant z - y \mid \boldeps_j) f_{\boldeps}(y)= \Phi(z-y, \Xi_p) \phi(y,\Xi_{\boldeps}).
\end{equation}

A simple application of dominated convergence theorem then implies, as $n_j \to \infty$, 
\begin{equation}\label{eq:convolutionnormal}
\P(\sqrt{n_j} (\hat \boldp_j - \tilde \boldp) \leqslant z) \to \int_{\mathbb{R}^{C-1}} \Phi(z-y, \Xi_p) \phi(y, \Xi_{\boldeps})\; dy.
\end{equation}

Note that the term on the right hand side of \cref{eq:convolutionnormal} is the distribution function of $\mathcal{N}_{C-1}(\bm 0, \Xi_{\boldp}+\Xi_{\boldeps})$ at the point $z \in \R^{C-1}$. Hence we arrive at the required result, \cref{eq:asymptotic-xnj}.
\end{proof}

\begin{proof}[\Cref{thm:posterior-consistency}]
Recall that the sigma field generated by the collection of the data $\{ \L_1, \L_2, \ldots,\L_K \}$ is $\F_K$, and the posterior means of $\boldmu$ and $\Sigma$, conditional on $\F_K$, are $\boldmu_{PM}$ and $\Sigma_{PM}$ respectively.   Then, \cref{eq:postconstmu} and \cref{eq:postconstSigma} can be equivalently stated as, 
\begin{equation}\label{eq:postconstmufilt}
\lim_{K \rightarrow \infty} \E\left[\Pi_K\left( \norm{\boldmu_{PM} - \boldmu_0} > \epsilon \mid \F_K \right)\right] =0,
\end{equation}
\begin{equation}\label{eq:postconstSigmafilt}
\lim_{K \rightarrow \infty} \E\left[\Pi_K\left( \norm{\Sigma_{PM} - \Sigma_0} > \epsilon \mid \F_K \right)\right] =0.
\end{equation}

We begin with the term

\begin{equation}\label{ineq:post_mu_twoparts}
\Pi_K\left( \norm{\boldmu_{PM} - \boldmu_0} > {\epsilon} \mid \F_K \right).
\end{equation} 

We know that $\boldmu_{PM}=\E\left(\left[\Sigma_p^{-1} + N\Sigma^{-1} \right]^{-1} \left[\Sigma_p^{-1} {\boldalpha} + \Sigma^{-1}\left(\sum_{i=1}^K n_i \L_i \right) \right]\mid \F_K\right)$, where the expectation is taken with respect to $\Sigma$. Since \cref{ineq:post_mu_twoparts} is either $1$ or $0$ depending on whether $\norm{\boldmu_{PM} - \boldmu_0} > \epsilon$ or not, on taking expectation over $\F_K$, we obtain 

\begin{equation}\label{ineq:expconstmu}
\P \left(\norm{\boldmu_{PM} - \boldmu_0} > {\epsilon}  \mid \F_K \right).
\end{equation}

To show that \cref{ineq:expconstmu} goes to $0$ as $K \rightarrow \infty$, we do the following: 

\begin{align}
\norm{\boldmu_{PM} - \boldmu_0} &= \norm{\E\left(\left[\Sigma_p^{-1} + N\Sigma^{-1} \right]^{-1} \left[\Sigma_p^{-1} {\boldalpha} + \Sigma^{-1}\left(\sum_{i=1}^K n_i \L_i \right) \right]\mid \F_K\right) - \boldmu_0} \nonumber \\
&\leqslant \E\left(\norm{\left[\Sigma_p^{-1} + N\Sigma^{-1} \right]^{-1} \left[\Sigma_p^{-1} {(\boldalpha-\boldmu_0)} + \Sigma^{-1}\left\{\sum_{i=1}^K n_i (\L_i-\boldmu_0) \right\} \right] }\mid \F_K \right) \nonumber \\
&\leqslant \E \left(\norm{\left[\Sigma_p^{-1} + N\Sigma^{-1} \right]^{-1}\Sigma_p^{-1}(\boldalpha-\boldmu_0)} \mid \F_k\right) \nonumber \\
\label{ineq:post_mu_pm}
& ~ \hspace{0.05in} \;  +\norm{\E\left(\left[\Sigma_p^{-1} + N\Sigma^{-1} \right]^{-1}\Sigma^{-1} \left\{\sum_{i=1}^K n_i (\L_i-\boldmu_0) \right\}\mid \F_K \right)}.
\end{align}

Applying the identity $(A+B)^{-1} = A^{-1} - A^{-1}(A^{-1}+B^{-1})^{-1}A^{-1}$ on the first term of the above,
\begin{equation}
\norm{\left[\Sigma_p^{-1} + N\Sigma^{-1} \right]^{-1}\Sigma_p^{-1}(\boldalpha-\boldmu_0)} = \norm{\left[\I - \Sigma_p\left(\Sigma_p + \frac{\Sigma}{N}\right)^{-1}\right](\boldalpha-\boldmu_0)},
\end{equation}
which goes to 0 a.e., using the properties of inverse Wishart distribution and the assumptions on the prior parameters we have. Consequently, $\E\norm{\left[\Sigma_p^{-1} + N\Sigma^{-1} \right]^{-1}\Sigma_p^{-1}(\boldalpha-\boldmu_0)}$ goes to 0 and hence, we can argue that 
\begin{equation}
\label{eq:convergence_post_mu_secondterm}
\P\left(\E\norm{\left[\Sigma_p^{-1} + N\Sigma^{-1} \right]^{-1}\Sigma_p^{-1}(\boldalpha-\boldmu_0)} > \frac{\epsilon}{2}\right) \to 0.
\end{equation}

For the second term in \cref{ineq:post_mu_pm}, first note that $\sum_{i=1}^K n_i (\L_i-\boldmu_0)$ is equivalent, in distributional sense, to $\sqrt{N}Z$, where $Z\sim \N(\bm{0},\Sigma_0)$. On the other hand,
\begin{equation}\label{ineq:postmu_variance_bound}
\norm{\E\left(\left[\Sigma_p^{-1} + N\Sigma^{-1} \right]^{-1}\Sigma^{-1}\right)} \leqslant \norm{\Sigma_p\E\left[\left(\Sigma + N\Sigma_p \right)^{-1}\right]} \leqslant  C_K\norm{\Sigma_p},
\end{equation} 
where $C_K=O(1/N)$. Using the above results, we can write
\begin{equation}
\P\left[\norm{\E\left(\left[\Sigma_p^{-1} + N\Sigma^{-1} \right]^{-1}\Sigma^{-1} \sqrt{N}Z \mid \F_K \right)} > \frac{\epsilon}{2}\right] \leqslant \P\left[\norm{C_K\sqrt{N}Z} > \frac{\epsilon}{4\norm{\Sigma_p}}\right].
\end{equation}

In light of the fact that $\norm{\Sigma_p}$ is finite and that $C_K=O(1/N)$, it is easy to see that the above probability goes to 0. This, along with \cref{eq:convergence_post_mu_secondterm}, completes the proof that \cref{ineq:expconstmu} goes to $0$ as $K \rightarrow \infty$.

In order to prove \cref{eq:postconstSigmafilt}, we follow a similar idea as above. Note that $\Psi$ is a constant positive definite matrix and $\nu$ is a finite constant. Following \cref{eq:Sigma-posterior}, for large $K$, straightforward calculations yield the following.
\begin{equation}\label{eq:sigma_pm_consistency_step1}
\norm{\Sigma_{PM}-\Sigma_0} = \frac1K \norm{\E\left[\sum_{i=1}^K n_(\L_i-\boldmu)(\L_i-\boldmu)^T - \Sigma_0\mid\F_k\right]}.
\end{equation}

Further, using the true value $\boldmu_0$ inside the term on the right hand side above, we get
\begin{align}
\sum_{i=1}^K n_i(\L_i-\boldmu)(\L_i-\boldmu)^T - \Sigma_0 &= \sum_{i=1}^K n_i(\L_i-\boldmu_0)(\L_i-\boldmu_0)^T - \Sigma_0 + N(\boldmu-\boldmu_0)(\boldmu-\boldmu_0)^T \nonumber \\
\label{eq:sigma_pm_consistency_step2}
& - 2\sum_{i=1}^K n_i(\L_i-\boldmu_0)(\boldmu-\boldmu_0)^T,
\end{align}

and subsequently,

\begin{align}
\P\left(\norm{\Sigma_{PM}-\Sigma_0} > \epsilon\mid\F_k\right) &\leqslant \P\left(\norm{\frac1K\sum_{i=1}^K n_i(\L_i-\boldmu_0)(\L_i-\boldmu_0)^T - \Sigma_0} > \frac{\epsilon}{3}\right) \nonumber \\
& + \P\left(\norm{\frac1K\E\left[N(\boldmu-\boldmu_0)(\boldmu-\boldmu_0)^T\mid\F_k\right]} > \frac{\epsilon}{3}\right) \nonumber \\
\label{ineq:postsigma_pm_threeparts}
& + \P\left(\frac2K\norm{\E\left[\sum_{i=1}^K n_i(\L_i-\boldmu_0)(\boldmu-\boldmu_0)^T\mid\F_k\right]} > \frac{\epsilon}{3}\right).
\end{align}

Using similar arguments as in \cref{ineq:postmu_variance_bound}, we can show that $\E[N(\boldmu-\boldmu_0)(\boldmu-\boldmu_0)^T \mid \F_k]$ is bounded and therefore, the second term goes to 0. Next, taking cue from the previous part regarding the consistency of $\boldmu$, it is also easy to see that for large $K$, $\E[n_i(\L_i-\boldmu_0)(\boldmu-\boldmu_0)^T\mid\F_k]$ is small enough. Hence, the third term in \cref{ineq:postsigma_pm_threeparts} also goes to 0. Finally, for the first term in that inequality, observe that $\sqrt{n_i}(\L_i-\boldmu_0) \sim \N_{C-1}(0,\Sigma_0)$ for all $i$. Thus, $W = \sum_{i=1}^K n_i(\L_i-\boldmu_0)(\L_i-\boldmu_0)^T \sim \mbox{Wishart}(\Sigma_0,K)$, and the law of large numbers implies that $W/K-\Sigma_0 \to 0$ in probability. Consequently, we get that $\P\left(\norm{\Sigma_{PM}-\Sigma_0} > \epsilon\right) \to 0$ and that completes the proof.
\end{proof}

\newpage
\setcounter{page}{1}

\section*{Appendix}

\appendix

\setcounter{table}{0}
\renewcommand{\thetable}{A\arabic{table}}
\setcounter{figure}{0}
\renewcommand{\thefigure}{A\arabic{figure}}

\section{Additional tables}\label{sec:additional-results}

\begin{table}[!ht]
\centering
\caption{Accuracy (in \%) of predicting the category with maximum count (corresponding final margins, averaged over all repetitions, are given in parentheses) for different values of $C,K,n$, when data are generated from SEO1.}
\label{tab:simulation-accuracy-margin-dgp1-detail}
\begin{tabular}{|llllll|}
\hline 
$C$ & $K$ & $n$   & Correct (average margin) & Incorrect (average margin) & No call (average margin) \\
\hline
3 & 25 & 100 & 95.8\% (1158) & 0.4\% (91) & 3.8\% (120) \\ 
3 & 25 & 1000 & 99.2\% (11894) & 0.2\% (255) & 0.6\% (399) \\ 
3 & 25 & 5000 & 98.62\% (54870) & 0.92\% (454) & 0.46\% (1108) \\ 
3 & 25 & 50000 & 100\% (549287) &  &  \\ 
\hline
3 & 50 & 100 & 95.75\% (2384) & 0.75\% (104) & 3.5\% (125) \\ 
3 & 50 & 1000 & 99\% (23425) &  & 1\% (346) \\ 
3 & 50 & 5000 & 99.5\% (114905) &  & 0.5\% (1581) \\ 
3 & 50 & 50000 & 100\% (1139338) &  &  \\ 
\hline
5 & 25 & 100 & 93.75\% (744) & 2.25\% (27) & 4\% (81) \\ 
5 & 25 & 1000 & 98.25\% (7395) & 0.25\% (38) & 1.5\% (368) \\ 
5 & 25 & 5000 & 98.75\% (38799) & 0.25\% (330) & 1\% (788) \\ 
5 & 25 & 50000 & 99.75\% (373575) &  & 0.25\% (2292) \\ 
\hline
5 & 50 & 100 & 94.5\% (1584) & 1.5\% (73) & 4\% (168) \\ 
5 & 50 & 1000 & 98.25\% (14388) & 0.5\% (516) & 1.25\% (378) \\ 
5 & 50 & 5000 & 99.5\% (71562) & 0.25\% (1461) & 0.25\% (858) \\ 
5 & 50 & 50000 & 100\% (730428) &  &  \\
\hline
\end{tabular}
\end{table}

\newpage

\begin{table}[!ht]
\centering
\caption{Accuracy (in \%) of predicting the category with maximum count (corresponding final margins, averaged over all repetitions, are given in parentheses) for different values of $C,K,n$, when data are generated from SEO2.}
\label{tab:simulation-accuracy-margin-dgp2-detail}
\begin{tabular}{|llllll|}
\hline 
$C$ & $K$ & $n$   & Correct (average margin) & Incorrect (average margin) & No call (average margin) \\
\hline
3 & 25 & 100 & 72.6\% (626) & 8\% (204) & 19.4\% (370) \\ 
3 & 25 & 1000 & 89.6\% (9374) & 2.8\% (2174) & 7.6\% (2866) \\ 
3 & 25 & 5000 & 89.6\% (54572) & 3.2\% (5723) & 7.2\% (8492) \\ 
3 & 25 & 50000 & 98.2\% (551259) & 1\% (20745) & 0.8\% (40335) \\ 
\hline
3 & 50 & 100 & 77.5\% (1227) & 6.75\% (355) & 15.75\% (480) \\ 
3 & 50 & 1000 & 88\% (18912) & 4\% (3189) & 8\% (3772) \\ 
3 & 50 & 5000 & 93.25\% (100611) & 3.75\% (7688) & 3\% (15393) \\ 
3 & 50 & 50000 & 96.75\% (1098195) & 1\% (41312) & 2.25\% (96337) \\ 
\hline
5 & 25 & 100 & 62.75\% (298) & 13.75\% (146) & 23.5\% (153) \\ 
5 & 25 & 1000 & 77.5\% (5681) & 8.75\% (1067) & 13.75\% (1758) \\ 
5 & 25 & 5000 & 85.75\% (33422) & 5.5\% (4861) & 8.75\% (7226) \\ 
5 & 25 & 50000 & 93\% (371349) & 2\% (32987) & 5\% (50911) \\ 
\hline
5 & 50 & 100 & 70\% (605) & 8.25\% (205) & 21.75\% (314) \\ 
5 & 50 & 1000 & 84.75\% (10986) & 4.75\% (1286) & 10.5\% (2891) \\ 
5 & 50 & 5000 & 91.75\% (67476) & 3.25\% (6624) & 5\% (12241) \\ 
5 & 50 & 50000 & 95.5\% (720905) & 2.25\% (71467) & 2.25\% (77728) \\
\hline
\end{tabular}
\end{table}

\newpage

\begin{table}[!ht]
\centering
\caption{Accuracy (in \%) of predicting the category with maximum count (corresponding final margins, averaged over all repetitions, are given in parentheses) for different values of $C,K,n$, when data are generated from SEO3.}
\label{tab:simulation-accuracy-margin-dgp3-detail}
\begin{tabular}{|llllll|}
\hline 
$C$ & $K$ & $n$   & Correct (average margin) & Incorrect (average margin) & No call (average margin) \\
\hline
3 & 25 & 100 & 43\% (230) & 23\% (123) & 34\% (163) \\ 
3 & 25 & 1000 & 49.8\% (3558) & 20.2\% (1948) & 30\% (2425) \\ 
3 & 25 & 5000 & 58.2\% (22490) & 15\% (10192) & 26.8\% (14920) \\ 
3 & 25 & 50000 & 73.8\% (323688) & 8.8\% (100474) & 17.4\% (187450) \\ 
\hline
3 & 50 & 100 & 43.2\% (414) & 18\% (258) & 38.8\% (297) \\ 
3 & 50 & 1000 & 56.2\% (6696) & 13.4\% (3131) & 30.4\% (4367) \\ 
3 & 50 & 5000 & 64.6\% (45409) & 7.2\% (16465) & 28.2\% (27686) \\ 
3 & 50 & 50000 & 79.4\% (700278) & 5.2\% (141302) & 15.4\% (270719) \\ 
\hline
5 & 25 & 100 & 40\% (91) & 26\% (64) & 34\% (47) \\ 
5 & 25 & 1000 & 42.6\% (1323) & 26.2\% (772) & 31.2\% (755) \\ 
5 & 25 & 5000 & 51\% (8658) & 24\% (5588) & 25\% (5423) \\ 
5 & 25 & 50000 & 65.6\% (154765) & 10.6\% (80310) & 23.8\% (81664) \\ 
\hline
5 & 50 & 100 & 42.4\% (154) & 22\% (106) & 35.6\% (110) \\ 
5 & 50 & 1000 & 45\% (2439) & 23.2\% (1423) & 31.8\% (1536) \\ 
5 & 50 & 5000 & 55.8\% (18972) & 14.2\% (9671) & 30\% (10822) \\
5 & 50 & 50000 & 70\% (305700) & 8.8\% (79016) & 21.2\% (184868) \\ 
\hline
\end{tabular}
\end{table}

\newpage

\begin{center}
\begin{longtable}{|lcc|cccc|c|}
\caption{Detailed results from all constituencies in Bihar.}\label{tab:bihar-fullresults} \\
\hline
\multicolumn{3}{|c|}{True data} & \multicolumn{4}{c|}{At the time of calling the race} & \\
Constituency & Rounds & Final margin & Round & Votes left (\%) & Lead & Predicted margin & Decision \\ 
\hline
\endhead
\hline
\multicolumn{8}{r@{}}{continued \ldots}\\
\endfoot
\hline
\endlastfoot
Agiaon & 28 & 48165 & 14 & 49.1 & 21498 & 42268 & Correct \\ 
Alamnagar & 37 & 29095 & 19 & 49.0 & 16498 & 32418 & Correct \\ 
Alauli & 25 & 2564 & 13 & 47.7 & 7899 & 14798 & Correct \\ 
Alinagar & 29 & 3370 & 28 & 0.8 & 3439 & 3450 & Correct \\ 
Amarpur & 31 & 3242 & 30 & 2.5 & 2341 & 2370 & Correct \\ 
Amnour & 29 & 3824 & 24 & 12.8 & 4556 & 5177 & Correct \\ 
Amour & 34 & 52296 & 17 & 49.3 & 28624 & 55829 & Correct \\ 
Araria & 35 & 47828 & 17 & 48.0 & 21851 & 41762 & Correct \\ 
Arrah & 36 & 3105 & 29 & 13.9 & 6758 & 7970 & Correct \\ 
Arwal & 30 & 19651 & 13 & 49.1 & 10167 & 19375 & Correct \\ 
Asthawan & 31 & 11530 & 16 & 46.9 & 9270 & 17362 & Correct \\ 
Atri & 34 & 7578 & 17 & 46.5 & 3218 & 5957 & Correct \\ 
Aurai & 33 & 48008 & 16 & 47.4 & 24134 & 45972 & Correct \\ 
Aurangabad & 34 & 2063 & 15 & 49.7 & 2754 & 5676 & Correct \\ 
Babubarhi & 32 & 12022 & 16 & 49.3 & 10414 & 20871 & Correct \\ 
Bachhwara & 32 & 737 & 15 & 49.7 & 5631 & 10832 & Incorrect \\ 
Bagaha & 33 & 30494 & 17 & 47.0 & 21112 & 39649 & Correct \\ 
Bahadurganj & 33 & 44978 & 17 & 47.4 & 14130 & 26583 & Correct \\ 
Bahadurpur & 34 & 2815 & 31 & 1.4 & 1965 & 1981 & Correct \\ 
Baikunthpur & 33 & 10805 & 18 & 43.1 & 5730 & 10274 & Correct \\ 
Baisi & 30 & 16312 & 15 & 46.6 & 6761 & 12218 & Incorrect \\ 
Bajpatti & 34 & 2325 & 29 & 11.2 & 4225 & 4778 & Correct \\ 
Bakhri & 28 & 439 &  &  &  &  & No call \\ 
Bakhtiarpur & 30 & 20694 & 15 & 47.9 & 7992 & 14667 & Correct \\ 
Balrampur & 35 & 53078 & 18 & 45.6 & 6706 & 12218 & Correct \\ 
Baniapur & 34 & 27219 & 17 & 48.7 & 19610 & 38052 & Correct \\ 
Banka & 27 & 17093 & 14 & 47.9 & 12133 & 23154 & Correct \\ 
Bankipur & 45 & 38965 & 22 & 49.3 & 22243 & 43644 & Correct \\ 
Banmankhi & 33 & 27872 & 16 & 47.7 & 11958 & 23007 & Correct \\ 
Barachatti & 33 & 6737 & 20 & 39.2 & 5916 & 9722 & Correct \\ 
Barari & 28 & 10847 & 14 & 46.7 & 26361 & 49359 & Correct \\ 
Barauli & 33 & 14493 & 16 & 47.1 & 4342 & 8183 & Correct \\ 
Barbigha & 25 & 238 &  &  &  &  & No call \\ 
Barh & 30 & 10084 & 21 & 28.5 & 7709 & 10450 & Correct \\ 
Barhara & 32 & 4849 & 28 & 12.1 & 5190 & 5932 & Correct \\ 
Barharia & 33 & 3220 & 16 & 50.0 & 8954 & 17765 & Correct \\ 
Baruraj & 30 & 43548 & 15 & 46.3 & 26826 & 49863 & Correct \\ 
Bathnaha & 32 & 47136 & 17 & 46.5 & 26047 & 48536 & Correct \\ 
Begusarai & 38 & 5392 & 17 & 50.0 & 4814 & 9165 & Incorrect \\ 
Belaganj & 32 & 23516 & 16 & 47.2 & 18889 & 35735 & Correct \\ 
Beldaur & 33 & 5289 & 16 & 48.4 & 5233 & 10518 & Correct \\ 
Belhar & 32 & 2713 & 17 & 48.6 & 6701 & 13164 & Correct \\ 
Belsand & 28 & 13685 & 13 & 49.2 & 6577 & 12898 & Correct \\ 
Benipatti & 32 & 32904 & 16 & 49.6 & 20784 & 41325 & Correct \\ 
Benipur & 31 & 6793 & 25 & 18.2 & 6957 & 8499 & Correct \\ 
Bettiah & 30 & 18375 & 17 & 46.8 & 4802 & 9154 & Correct \\ 
Bhabua & 29 & 9447 & 15 & 48.8 & 4971 & 9665 & Correct \\ 
Bhagalpur & 36 & 950 & 20 & 40.6 & 9662 & 15892 & Correct \\ 
Bhorey & 37 & 1026 & 19 & 48.5 & 4141 & 7857 & Correct \\ 
Bibhutipur & 30 & 40369 & 14 & 49.4 & 17213 & 33918 & Correct \\ 
Bihariganj & 34 & 19459 & 18 & 45.0 & 7497 & 13702 & Correct \\ 
Biharsharif & 41 & 15233 & 19 & 47.5 & 7362 & 14029 & Correct \\ 
Bihpur & 30 & 6348 & 16 & 46.6 & 9088 & 16985 & Correct \\ 
Bikram & 35 & 35390 & 17 & 49.9 & 21569 & 42937 & Correct \\ 
Bisfi & 36 & 10469 & 19 & 45.6 & 8169 & 14305 & Correct \\ 
Bochaha & 30 & 11615 & 15 & 48.6 & 7128 & 13904 & Correct \\ 
Bodh Gaya & 34 & 4275 & 17 & 47.0 & 2704 & 5174 & Incorrect \\ 
Brahampur & 37 & 50537 & 18 & 47.4 & 23505 & 44079 & Correct \\ 
Buxar & 30 & 3351 & 28 & 6.8 & 2508 & 2680 & Correct \\ 
Chainpur & 36 & 23650 & 18 & 48.2 & 10926 & 21057 & Correct \\ 
Chakai & 32 & 654 & 17 & 47.5 & 3278 & 6470 & Correct \\ 
Chanpatia & 27 & 13680 & 14 & 47.2 & 8080 & 15260 & Correct \\ 
Chapra & 36 & 7222 & 30 & 14.6 & 6135 & 7234 & Correct \\ 
Chenari & 33 & 17489 & 16 & 49.6 & 11612 & 23038 & Correct \\ 
Cheria Bariarpur & 29 & 40379 & 15 & 47.4 & 19725 & 37310 & Correct \\ 
Chhatapur & 32 & 20858 & 16 & 48.6 & 15892 & 30405 & Correct \\ 
Chiraiya & 31 & 17216 & 15 & 49.4 & 10055 & 20033 & Correct \\ 
Danapur & 38 & 16005 & 19 & 47.9 & 18866 & 35408 & Correct \\ 
Darauli & 34 & 11771 & 16 & 49.8 & 7508 & 14898 & Correct \\ 
Daraundha & 35 & 11492 & 20 & 41.4 & 4388 & 7409 & Correct \\ 
Darbhanga Rural & 30 & 2019 & 14 & 49.6 & 6509 & 13277 & Correct \\ 
Darbhanga & 32 & 10870 & 16 & 47.4 & 11923 & 22419 & Correct \\ 
Dehri & 32 & 81 & 15 & 46.7 & 5497 & 10371 & Correct \\ 
Dhaka & 34 & 10396 & 17 & 47.5 & 18059 & 33826 & Correct \\ 
Dhamdaha & 35 & 33701 & 18 & 48.1 & 24815 & 47600 & Correct \\ 
Dhauraiya & 31 & 2687 & 15 & 49.8 & 9155 & 18294 & Incorrect \\ 
Digha & 51 & 46073 & 25 & 49.3 & 23335 & 45542 & Correct \\ 
Dinara & 32 & 7896 & 16 & 47.9 & 4136 & 8036 & Incorrect \\ 
Dumraon & 34 & 23854 & 17 & 47.4 & 14094 & 26750 & Correct \\ 
Ekma & 33 & 13683 & 17 & 48.0 & 5019 & 9433 & Correct \\ 
Fatuha & 29 & 19407 & 15 & 47.1 & 11926 & 21926 & Correct \\ 
Forbesganj & 36 & 19749 & 18 & 49.3 & 9320 & 18607 & Correct \\ 
Gaighat & 33 & 7345 & 27 & 17.4 & 4232 & 5180 & Correct \\ 
Garkha & 33 & 9746 & 19 & 39.7 & 6592 & 10794 & Correct \\ 
Gaura Bauram & 25 & 7519 & 13 & 47.7 & 8065 & 15274 & Correct \\ 
Gaya Town & 29 & 12123 & 15 & 48.7 & 9590 & 18989 & Correct \\ 
Ghosi & 28 & 17804 & 14 & 49.3 & 17071 & 33251 & Correct \\ 
Gobindpur & 35 & 32776 & 18 & 48.3 & 21332 & 41204 & Correct \\ 
Goh & 32 & 35377 & 17 & 47.5 & 15573 & 29875 & Correct \\ 
Gopalganj & 37 & 36641 & 18 & 48.6 & 16397 & 31998 & Correct \\ 
Gopalpur & 30 & 24580 & 15 & 49.3 & 9626 & 19022 & Correct \\ 
Goriakothi & 35 & 12345 & 18 & 47.3 & 10776 & 20645 & Correct \\ 
Govindganj & 28 & 27924 & 14 & 46.7 & 10488 & 19611 & Correct \\ 
Gurua & 31 & 6152 & 16 & 47.9 & 9515 & 18353 & Correct \\ 
Hajipur & 35 & 3248 & 19 & 47.7 & 6633 & 12491 & Correct \\ 
Harlakhi & 32 & 17815 & 15 & 49.6 & 9927 & 19681 & Correct \\ 
Harnaut & 33 & 27050 & 16 & 48.4 & 11151 & 21626 & Correct \\ 
Harsidhi & 28 & 16071 & 14 & 49.7 & 11104 & 21881 & Correct \\ 
Hasanpur & 30 & 21039 & 18 & 39.7 & 8936 & 14705 & Correct \\ 
Hathua & 32 & 30237 & 16 & 48.1 & 18605 & 35616 & Correct \\ 
Hayaghat & 25 & 10420 & 23 & 6.4 & 5735 & 6090 & Correct \\ 
Hilsa & 33 & 13 & 20 & 33.7 & 4433 & 6678 & Incorrect \\ 
Hisua & 41 & 16775 & 20 & 48.1 & 15057 & 28784 & Correct \\ 
Imamganj & 32 & 16177 & 16 & 48.5 & 7363 & 13978 & Correct \\ 
Islampur & 31 & 3767 & 29 & 4.8 & 3297 & 3474 & Correct \\ 
Jagdishpur & 32 & 21492 & 16 & 47.8 & 8949 & 17206 & Correct \\ 
Jale & 33 & 21926 & 17 & 47.3 & 6953 & 13395 & Correct \\ 
Jamalpur & 33 & 4468 & 21 & 36.2 & 4626 & 7197 & Correct \\ 
Jamui & 32 & 41009 & 17 & 47.9 & 25100 & 48494 & Correct \\ 
Jehanabad & 33 & 33399 & 16 & 47.8 & 17010 & 32635 & Correct \\ 
Jhajha & 35 & 1779 & 17 & 47.3 & 8142 & 15372 & Correct \\ 
Jhanjharpur & 34 & 41861 & 17 & 49.7 & 28969 & 57628 & Correct \\ 
Jokihat & 31 & 7543 & 18 & 39.8 & 3951 & 6575 & Correct \\ 
Kadwa & 30 & 31919 & 15 & 47.4 & 13917 & 26993 & Correct \\ 
Kahalgaon & 36 & 42947 & 18 & 47.5 & 26297 & 49848 & Correct \\ 
Kalyanpur (A) & 34 & 10329 & 17 & 49.0 & 9814 & 18960 & Correct \\ 
Kalyanpur (B) & 27 & 852 & 13 & 48.1 & 4335 & 8519 & Correct \\ 
Kanti & 33 & 10254 & 25 & 23.2 & 5588 & 7257 & Correct \\ 
Karakat & 36 & 17819 & 17 & 48.8 & 5691 & 11004 & Correct \\ 
Kargahar & 35 & 3666 & 32 & 5.0 & 3995 & 4226 & Correct \\ 
Kasba & 31 & 17081 & 21 & 32.2 & 11768 & 17215 & Correct \\ 
Katihar & 29 & 11183 & 18 & 37.8 & 6557 & 10741 & Correct \\ 
Katoria & 27 & 6704 & 24 & 8.6 & 5688 & 6253 & Correct \\ 
Keoti & 31 & 5267 & 26 & 13.3 & 6814 & 7892 & Correct \\ 
Kesaria & 29 & 9352 & 14 & 48.8 & 5048 & 9796 & Correct \\ 
Khagaria & 28 & 2661 & 24 & 9.1 & 3550 & 3921 & Correct \\ 
Khajauli & 31 & 23037 & 16 & 48.9 & 18093 & 35361 & Correct \\ 
Kishanganj & 33 & 1221 & 16 & 48.5 & 7588 & 14366 & Correct \\ 
Kochadhaman & 26 & 36072 & 14 & 47.1 & 25184 & 46871 & Correct \\ 
Korha & 30 & 29007 & 15 & 47.7 & 23079 & 43352 & Correct \\ 
Kuchaikote & 36 & 20753 & 17 & 49.4 & 12435 & 24464 & Correct \\ 
Kumhrar & 49 & 26466 & 26 & 49.4 & 25992 & 51127 & Correct \\ 
Kurhani & 33 & 480 & 15 & 50.0 & 6379 & 12687 & Correct \\ 
Kurtha & 26 & 27542 & 13 & 49.6 & 12732 & 25261 & Correct \\ 
Kusheshwarasthan & 27 & 7376 & 14 & 47.1 & 5444 & 10264 & Correct \\ 
Kutumba & 28 & 16330 & 15 & 46.7 & 9879 & 18414 & Correct \\ 
Lakhisarai & 45 & 10709 & 22 & 48.0 & 5317 & 10227 & Correct \\ 
Lalganj & 34 & 26613 & 17 & 47.9 & 18727 & 36231 & Correct \\ 
Laukaha & 37 & 9471 & 18 & 48.2 & 10665 & 20498 & Correct \\ 
Lauriya & 27 & 29172 & 14 & 47.9 & 14345 & 27468 & Correct \\ 
Madhepura & 35 & 15072 & 16 & 49.6 & 6466 & 12643 & Correct \\ 
Madhuban & 27 & 6115 & 18 & 34.1 & 4790 & 7219 & Correct \\ 
Madhubani & 37 & 6490 & 18 & 48.7 & 7292 & 13924 & Correct \\ 
Maharajganj & 34 & 1638 & 31 & 3.1 & 1271 & 1291 & Correct \\ 
Mahishi & 31 & 1972 & 16 & 48.0 & 9344 & 17525 & Incorrect \\ 
Mahnar & 31 & 7781 & 16 & 46.9 & 4244 & 7697 & Correct \\ 
Mahua & 30 & 13687 & 15 & 48.6 & 7468 & 14255 & Correct \\ 
Makhdumpur & 26 & 21694 & 14 & 46.5 & 13559 & 24842 & Correct \\ 
Maner & 35 & 32919 & 18 & 47.6 & 24582 & 46654 & Correct \\ 
Manihari & 31 & 20679 & 14 & 49.3 & 9831 & 19437 & Correct \\ 
Manjhi & 32 & 25154 & 16 & 49.1 & 5858 & 11338 & Correct \\ 
Marhaura & 28 & 10966 & 18 & 34.7 & 3863 & 5834 & Correct \\ 
Masaurhi & 37 & 32161 & 19 & 48.7 & 16543 & 32147 & Correct \\ 
Matihani & 38 & 65 & 19 & 48.3 & 5017 & 9799 & Incorrect \\ 
Minapur & 30 & 15321 & 15 & 48.8 & 6930 & 13343 & Correct \\ 
Mohania & 30 & 11100 & 15 & 46.3 & 5643 & 10447 & Correct \\ 
Mohiuddinnagar & 29 & 15195 & 21 & 23.6 & 6799 & 9229 & Correct \\ 
Mokama & 31 & 35634 & 16 & 48.1 & 21745 & 41507 & Correct \\ 
Morwa & 29 & 10550 & 16 & 40.9 & 5802 & 9608 & Correct \\ 
Motihari & 33 & 14987 & 18 & 49.8 & 11764 & 23447 & Correct \\ 
Munger & 37 & 1346 & 36 & 0.4 & 885 & 888 & Correct \\ 
Muzaffarpur & 34 & 6132 & 20 & 41.9 & 5124 & 8643 & Correct \\ 
Nabinagar & 30 & 19926 & 14 & 49.1 & 14704 & 28899 & Correct \\ 
Nalanda & 32 & 15878 & 15 & 49.6 & 5022 & 9756 & Correct \\ 
Narkatia & 30 & 27377 & 15 & 49.1 & 16619 & 32266 & Correct \\ 
Narkatiaganj & 28 & 21519 & 14 & 48.1 & 11391 & 22076 & Correct \\ 
Narpatganj & 34 & 28681 & 16 & 49.5 & 12358 & 24689 & Correct \\ 
Nathnagar & 36 & 7481 & 26 & 24.8 & 10089 & 13833 & Correct \\ 
Nautan & 30 & 26106 & 16 & 46.6 & 15003 & 28245 & Correct \\ 
Nawada & 37 & 25835 & 19 & 42.7 & 6251 & 10608 & Correct \\ 
Nirmali & 31 & 44195 & 16 & 49.7 & 20916 & 41440 & Correct \\ 
Nokha & 31 & 17212 & 16 & 48.2 & 9007 & 17107 & Correct \\ 
Obra & 32 & 22233 & 16 & 49.9 & 8742 & 17256 & Correct \\ 
Paliganj & 30 & 30928 & 15 & 48.0 & 17531 & 33748 & Correct \\ 
Parbatta & 33 & 1178 & 32 & 0.2 & 1511 & 1511 & Correct \\ 
Parihar & 32 & 1729 & 16 & 48.6 & 3935 & 7704 & Correct \\ 
Paroo & 32 & 14722 & 16 & 49.5 & 4446 & 8661 & Correct \\ 
Parsa & 27 & 16947 & 14 & 46.7 & 8660 & 15974 & Correct \\ 
Patepur & 31 & 25958 & 16 & 47.7 & 13250 & 25623 & Correct \\ 
Patna Sahib & 39 & 18281 & 21 & 48.9 & 8388 & 16381 & Correct \\ 
Phulparas & 36 & 11198 & 18 & 49.3 & 8810 & 17390 & Correct \\ 
Phulwari & 38 & 13870 & 20 & 47.6 & 12749 & 24323 & Correct \\ 
Pipra (A) & 36 & 8605 & 18 & 49.6 & 6769 & 13534 & Correct \\ 
Pipra (B) & 30 & 19716 & 15 & 48.1 & 7511 & 14380 & Correct \\ 
Pirpainti & 35 & 26994 & 18 & 48.9 & 17639 & 34501 & Correct \\ 
Pranpur & 32 & 3266 & 16 & 46.7 & 10395 & 18800 & Correct \\ 
Purnia & 35 & 32288 & 20 & 48.3 & 24222 & 45472 & Correct \\ 
Rafiganj & 34 & 9219 & 17 & 48.4 & 10410 & 20108 & Correct \\ 
Raghopur & 38 & 37820 & 19 & 47.3 & 9179 & 17353 & Correct \\ 
Raghunathpur & 33 & 17578 & 15 & 48.6 & 13250 & 25591 & Correct \\ 
Rajapakar & 28 & 1497 & 27 & 0.3 & 1728 & 1744 & Correct \\ 
Rajauli & 37 & 12166 & 19 & 47.2 & 5413 & 10246 & Correct \\ 
Rajgir & 33 & 16132 & 17 & 46.5 & 5070 & 9708 & Correct \\ 
Rajnagar & 34 & 19388 & 17 & 49.0 & 9794 & 19390 & Correct \\ 
Rajpur & 34 & 20565 & 17 & 48.8 & 11750 & 23016 & Correct \\ 
Ramgarh & 32 & 305 & 15 & 47.8 & 4488 & 8404 & Correct \\ 
Ramnagar & 32 & 16087 & 16 & 47.9 & 11121 & 21162 & Correct \\ 
Raniganj & 36 & 2395 & 34 & 1.5 & 2550 & 2583 & Correct \\ 
Raxaul & 28 & 37094 & 14 & 49.4 & 23282 & 45791 & Correct \\ 
Riga & 33 & 32851 & 18 & 48.7 & 19030 & 37330 & Correct \\ 
Rosera & 35 & 35814 & 19 & 47.0 & 30151 & 56895 & Correct \\ 
Runnisaidpur & 30 & 24848 & 16 & 46.4 & 22558 & 41764 & Correct \\ 
Rupauli & 33 & 19343 & 16 & 50.0 & 15912 & 31754 & Correct \\ 
Saharsa & 38 & 20177 & 20 & 49.3 & 23357 & 46054 & Correct \\ 
Sahebganj & 32 & 15393 & 16 & 48.4 & 5297 & 10446 & Correct \\ 
Sahebpur Kamal & 26 & 13846 & 13 & 48.7 & 9133 & 17390 & Correct \\ 
Sakra & 28 & 1742 & 18 & 30.7 & 4361 & 6242 & Correct \\ 
Samastipur & 30 & 4588 & 26 & 9.1 & 5765 & 6425 & Correct \\ 
Sandesh & 30 & 50109 & 16 & 48.0 & 26711 & 51168 & Correct \\ 
Sarairanjan & 31 & 3722 & 15 & 49.6 & 3992 & 8100 & Correct \\ 
Sasaram & 36 & 25779 & 20 & 48.3 & 19372 & 37283 & Correct \\ 
Shahpur & 36 & 22384 & 18 & 49.6 & 5654 & 11182 & Correct \\ 
Sheikhpura & 28 & 6003 & 21 & 20.6 & 4134 & 5195 & Correct \\ 
Sheohar & 32 & 36461 & 16 & 47.7 & 16516 & 31584 & Correct \\ 
Sherghati & 30 & 16449 & 15 & 49.8 & 10533 & 21095 & Correct \\ 
Sikandra & 32 & 5668 & 21 & 34.4 & 5237 & 7854 & Correct \\ 
Sikta & 30 & 2080 & 21 & 24.5 & 4218 & 5570 & Incorrect \\ 
Sikti & 29 & 13716 & 14 & 49.5 & 15465 & 29903 & Correct \\ 
Simri Bakhtiarpur & 36 & 1470 & 34 & 1.0 & 1995 & 2017 & Correct \\ 
Singheshwar & 33 & 4995 & 29 & 6.8 & 5231 & 5640 & Correct \\ 
Sitamarhi & 31 & 11946 & 27 & 8.7 & 7984 & 8768 & Correct \\ 
Siwan & 36 & 1561 & 16 & 48.1 & 5874 & 11223 & Incorrect \\ 
Sonbarsha & 32 & 13732 & 17 & 46.6 & 4935 & 9348 & Correct \\ 
Sonepur & 29 & 6557 & 15 & 47.3 & 4988 & 9105 & Incorrect \\ 
Sugauli & 30 & 3045 & 16 & 48.6 & 4810 & 9193 & Correct \\ 
Sultanganj & 35 & 11603 & 28 & 16.6 & 6231 & 7414 & Correct \\ 
Supaul & 29 & 28246 & 15 & 49.2 & 10219 & 20448 & Correct \\ 
Surajgarha & 40 & 9327 & 19 & 50.0 & 4647 & 9650 & Incorrect \\ 
Sursand & 35 & 9242 & 17 & 49.0 & 8237 & 15987 & Correct \\ 
Taraiya & 31 & 11542 & 15 & 49.9 & 9042 & 18196 & Correct \\ 
Tarapur & 34 & 7256 & 22 & 31.1 & 5871 & 8452 & Correct \\ 
Tarari & 36 & 10598 & 21 & 33.4 & 3105 & 4683 & Correct \\ 
Teghra & 31 & 47495 & 16 & 47.8 & 19584 & 37529 & Correct \\ 
Thakurganj & 31 & 23509 & 16 & 49.3 & 11008 & 21202 & Correct \\ 
Tikari & 35 & 2745 & 17 & 49.7 & 5312 & 10259 & Incorrect \\ 
Triveniganj & 29 & 3402 & 28 & 2.8 & 2499 & 2552 & Correct \\ 
Ujiarpur & 32 & 23010 & 16 & 48.5 & 16022 & 30998 & Correct \\ 
Vaishali & 34 & 7629 & 17 & 48.0 & 5278 & 10524 & Correct \\ 
Valmikinagar & 33 & 21825 & 16 & 47.0 & 15531 & 29394 & Correct \\ 
Warisnagar & 35 & 13913 & 17 & 47.9 & 6403 & 12099 & Correct \\ 
Warsaliganj & 38 & 9073 & 19 & 49.2 & 8982 & 17503 & Correct \\ 
Wazirganj & 34 & 22422 & 16 & 49.6 & 12787 & 25470 & Correct \\ 
Ziradei & 30 & 25156 & 15 & 47.0 & 13708 & 25846 & Correct \\ 
\hline
\end{longtable}
\end{center}

\end{document}